\documentclass[11pt]{article}
\usepackage{amssymb,amsmath,amsthm}
\usepackage{cancel}
\usepackage{graphics}
\usepackage{diagbox}
\usepackage{graphicx}
\usepackage{caption}
\usepackage{hyperref}
\usepackage{enumerate}
\usepackage{color}         
\usepackage{overpic}
\usepackage{subfig}
\usepackage{algorithm}
\usepackage{algpseudocode}
\usepackage{cancel}
\usepackage{thmtools, thm-restate}

\DeclareMathOperator*{\argmin}{arg~min}

\DeclareMathOperator{\supp}{supp}

\newcommand{\bfx}{{\bf x}}

\newcommand{\bfu}{{\bf u}}
\newcommand{\bfv}{{\bf v}}

\newcommand{\bfy}{{\bf y}}
\newcommand{\bfe}{{\bf e}}
\newcommand{\bfz}{{\bf z}}
\newcommand{\bfw}{{\bf w}}
\newcommand{\bfr}{{\bf r}}
\newcommand{\bfL}{{\bf L}}
\newcommand{\bfX}{{\bf X}}
\newcommand{\bfY}{{\bf Y}}
\newcommand{\bfE}{{\bf E}}
\newcommand{\bfA}{{\bf A}}

\newcommand{\bfone}{{\bf 1}}

\def\3bar{{|\hspace{-.02in}|\hspace{-.02in}|}}

\theoremstyle{plain} 
\newtheorem{theorem}{Theorem}[section]
\newtheorem*{theorem*}{Theorem}
\newtheorem{lemma}[theorem]{Lemma}

\theoremstyle{definition}
\newtheorem{remark}[theorem]{Remark}
\newtheorem{definition}[theorem]{Definition}

\numberwithin{equation}{section}

\begin{document}
\algdef{SE}[SUBALG]{Indent}{EndIndent}{}{\algorithmicend\ }%
\algtext*{Indent}
\algtext*{EndIndent}

\title{A Compressive Sensing Approach to Community Detection\\ with Applications}
\author{Ming-Jun Lai\footnote{mjlai@uga.edu. Department of Mathematics,
University of Georgia, Athens, GA 30602.
This research is partially supported by 
the National Science Foundation under the grant \#DMS 1521537. }
\and
Daniel Mckenzie \footnote{mckenzie@math.uga.edu. Department of 
Mathematics, University of Georgia, Athens, GA 30602. The second author gratefully acknowledges the financial support of the NRF of South Africa while this research was conducted}}
\maketitle
\begin{abstract}
The community detection problem for graphs asks one to partition the $n$
 vertices $V$ of a graph $G$ into $k$ communities, or clusters, such 
that there are many intracluster edges and few intercluster edges. Of 
course this is equivalent to finding a permutation matrix $\mathbf{P}$ such 
that, if $\bfA$ denotes the adjacency matrix of $G$, then $\mathbf{P}\bfA\mathbf{P}^\top $ is 
approximately block diagonal. As there are $k^n$ possible partitions of 
$n$ vertices into $k$ subsets, directly determining the optimal 
clustering is clearly infeasible. Instead one seeks to solve a more 
tractable approximation to the clustering problem. In this paper 
we reformulate the community detection problem via sparse solution of 
a linear system associated with the Laplacian of a graph $G$ 
and then develop a two-stage approach based on a thresholding technique and a  
compressive sensing algorithm to find a sparse solution which 
corresponds to the community containing a vertex of 
interest in $G$. Crucially, our approach results in an algorithm which is able 
to find a single cluster of size $n_0$ in $\mathcal{O}(n\ln(n) n_0)$ operations and all $k$ clusters in fewer than 
$\mathcal{O}(n^2\ln(n))$ operations. This is a marked 
improvement over the classic spectral clustering algorithm, which is 
unable to find a single 
cluster at a time and takes approximately $\mathcal{O}(n^3)$ operations 
to find all $k$ clusters. Moreover, we are able to provide robust 
guarantees of success for the case where $G$ is drawn at random from 
the Stochastic Block Model, a popular model for graphs with clusters. 
Extensive numerical results are also provided, showing the efficacy 
of our algorithm on both synthetic and real-world data sets.
\end{abstract}

\section{Introduction}
The  clustering problem for a graph $G = (V,E)$ is to 
divide the vertex set $V$ into subsets $V = C_{1}\cup \ldots \cup C_{k}$
 such that there are many intracluster edges 
(edges between vertices in the same cluster) and few intercluster edges 
(edges between vertices in different clusters) in $E$. This is a widely studied problem in exploratory data analysis, as one can reasonably assume that vertices in the same cluster a `similar', in some sense.  We refer the reader to the survey article \cite{Fortunato2010} for further details and a thorough overview of existing algorithmic approaches. We note that \cite{Fortunato2010} refers to the clustering problem as the community detection problem, and we shall use these two phrases interchangeably.  As is well-known, detecting clusters in $G$ is equivalent to finding a permutation matrix $\mathbf{P}$ such that if $\bfA$ is the adjacency matrix of $G$, then $\mathbf{P}\bfA\mathbf{P}^\top $ is almost block diagonal. Thus, we can think of the clustering problem as a special case of the matrix reduction problem where the matrix in question has binary entries. 

One class of robust and accurate algorithms used to solve the clustering problem are the spectral algorithms. Loosely, they work as follows. Suppose that $|V| = n$ and let $\bfL$ denote the graph Laplacian, while $\bfone_{C_a}$ denotes the indicator vector of the $a$-th cluster (both to be defined in \S \ref{section:GraphTheory}). Suppose further that it is known \emph{a priori} that $G$ has $k$ clusters. Let $\bfv_{1},\ldots, \bfv_{k}$ be orthogonal eigenvectors associated to the $k$ smallest eigenvalues of $L$, and consider the subspace $\text{span}\{\bfv_1,\ldots, \bfv_k\}\subset \mathbb{R}^{n}$. One can show that, under certain conditions, this subspace is `close' to $\text{span}\{\bfone_{C_1},\ldots, 
\bfone_{C_k}\}$ and hence one can use the basis $\{\bfv_1,\ldots, \bfv_k\}$ to infer the supports of the basis  $\{\bfone_{C_1},\ldots, \bfone_{C_k}\}$ , thus determining the 
clusters $C_{1}, \ldots, C_{k}$ (of course $\supp(\bfone_{C_a}) = C_a$). We refer the reader to \cite{Ng2001}, \cite{Luxburg2007} or \cite{Nascimento2011} for details. \\

Despite its theoretical and experimental success,the spectral approach has three main drawbacks:
\begin{enumerate}
\item The number of clusters $k$, needs to be known \emph{a priori}. Clearly this is not always the case for real data sets. 
\item 
The algorithm cannot be used to find only a few clusters. As in many applications one is only interested in finding one or two clusters ( thinking of the problem of identifying friends or associates of a given user from a social network data set). Moreover in other cases where the data set is extremely large, or only partially known, it might be computationally infeasible to identify all clusters.
\item Computing an eigen decomposition of $\bfL$ typically requires $\mathcal{O}(n^3)$ operations, making spectral methods prohibitively slow for truly large data sets, such as those arising from electronic social networks like Facebook or LinkedIn, or those arising from problems in Machine Learning.
\end{enumerate}
 
 However, in many situations (the social network example being one such case), the expected size of the clusters, $n_0$, is small compared to $n$, and hence the indicator vectors $\mathbf{1}_{C_a}$ will be \emph{sparse}. Our approach is to adapt sparse recovery algorithms from the compressive 
sensing study to solve the following problem:
\begin{equation}
\argmin ||\bfL\bfx ||_{2}: \quad \text{ subject to } x_{i} = 1 \text{ and } ||
\bfx ||_{0} \leq n_0\label{eq: CompClustering}
\end{equation}
to determine, directly, an approximation to the indicator vector of the cluster containing the vertex of interest $\bfv_i$. 
One can then recover the cluster by considering the support of this vector. If desired, one can iterate the algorithm to find all the remaining clusters of $G$. 

Adapting compressive sensing algorithms to solve \eqref{eq: 
CompClustering} proves challenging, as $\bfL$ is a poorly conditioned 
sensing matrix.  In general, a greedy type algorithm such as orthogonal 
matching pursuit or iterative hard thresholding (cf. 
\cite{Foucart2013}) work very well when there are no intercluster 
edges (in this case finding clusters reduces to finding connected components). 
Unfortunately, in the presence of even a small number of intercluster 
edges, the first few iterations of a greedy algorithm are likely to 
pick some indices outside the desired cluster. To overcome this 
difficulty, we propose a novel two stage algorithm  (see Algorithm~\ref
{algorithm:CompClust} in \S 8) in which the first stage identifies a 
subset $\Omega \subset V$ which contains the cluster of interest with 
high probability. The second stage then extracts the cluster of 
interest from $\Omega$ using a greedy algorithm (we use Subspace 
Pursuit cf. \cite{Dai2009}). In addition to the aforementioned 
algorithm, the main contributions of this paper are the following: 
\begin{enumerate}
\item An analysis of the Restricted Isometry and Coherence properties 
of the graph Laplacian.  In particular, we provide a series of 
probabilistic bounds on the restricted isometry constants and coherence 
of Laplacians of graphs drawn from a well-studied model of random 
graphs (the Stochastic Block Model (SBM)). See \S \ref{section:RIPLaplacian} and 
\S \ref{section:CoherenceProperties}. 
\item A proof that the Optimal Matching Pursuit (OMP) algorithm 
can be successfully used to detect connected 
components of any graph, by solving \eqref{eq: CompClustering}  
(see \S \ref{section:CompClustforq_zero}.) 
\item A proof that, given a vertex $\bfv_i$, our Single Cluster Pursuit (SCP)   
Algorithm~\ref{algorithm:CompClust}  successfully finds the cluster 
containing $\bfv_i$ when $G$ is drawn from the Stochastic Block Model (SBM), 
for a certain range of parameters and with probability tending to $1$ as $n\to\infty$. We 
achieve this by combining the bounds of contribution  with the theory 
of totally perturbed compressive sensing, e.g. in \cite{Herman2010}.
\item An analysis of the computational complexity of Algorithm \ref{algorithm:CompClust}, showing that it finds a single cluster in 
$\mathcal{O}(n^2  +\ln(n)nn_0)$ time and all clusters in $\mathcal{O}(kn^2 + \ln(n)n^2)$ time.
\end{enumerate}
        
The structure of this paper is as follows. After briefly reviewing 
related work in \S 2, in \S 3 and \S 4 we acquaint the reader with the 
necessary concepts from spectral graph theory and  compressive sensing,
respectively. In \S 5 and \S 6 we study the restricted isometry property
and coherence property of Laplacians of random graphs. In \S 7 and 
\S 8, we describe two algorithms to handle graphs from 
the Stochastic Block Model ${\cal G}(n, k, p, q)$ for $q=0$ and $q>0$, 
with respectively. 
We shall show that, under certain mild assumptions on the graph $G$, 
the algorithms will indeed find the correct clustering with high 
probability.  \S 9 contains the computational 
complexity analysis, some possible extensions and several future 
research directions. Finally, in \S 10 we present the results of 
several numerical experiments to demonstrate the accuracy and speed of 
our algorithm and its powerful performance. 

\section{Related Work}
The notion of community detection in graphs arises independently in 
multiple fields of applied science, such as Sociology (\cite{Zachary1977},\cite{Newman2002}), Computer Engineering (\cite{Hagen1992}), Machine Learning (\cite{Shi2000}) and Bioinformatics (\cite{Cline2007}). In addition, many data sets can be represented as graphs by considering data points as vertices and attaching edges between vertices that are `close' with respect to an appropriate metric. Thus, community detection algorithms can also be used to detect clusters in general data sets, and indeed they have been shown to be superior to other clustering algorithms (for example $k$-means) at detecting non-convex clusters (\cite{Jain2010}). \\
The canonical probabilistic model of a graph containing communities is the Stochastic Block Model (SBM), first 
explicitly introduced in the literature in the early 1980's by 
Holland,  Laskey, and Leinhardt in \cite{Holland1983}. Since then, there has been an explosion of interest in the SBM, driven in large part by its many applications. In \cite{Abbe2017}, Abbe identifies three forms of the community detection problem, based on the kind of asymptotic accuracy we require (here, as in the rest of the paper, when we speak of asymptotics we are considering the situation where the number of vertices, $n$, goes to $\infty$). In this paper we shall be concerned with the \emph{Exact Recovery Problem}, where we require that $\mathbb{P}\left( \# \text{misclassified vertices} > 0\right) = o(1)$. A fundamental information theoretic barrier to exact recovery is given by the following result (cf. 
\cite{Abbe2015}):

\begin{theorem}
The exact recovery problem , i.e. $\mathbb{P}\left( \# \text{misclassified vertices} > 0\right) = o(1)$ with respect to $n$ for the symmetric SBM $\mathcal{G}(n,k,p,q)$ is solvable in polynomial time if, 
writing $p = P\ln(n)/n$ and $q = Q\ln(n)/n$:
\begin{equation}
\frac{1}{k}\left(\sqrt{P} - \sqrt{Q}\right) > 1
\label{eq:DetectionBound}
\end{equation}
and not solvable if:
\begin{equation}
\frac{1}{k}\left(\sqrt{P} - \sqrt{Q}\right) < 1
\end{equation}
\end{theorem}

\begin{remark}
This result was proved in \cite{Abbe2015} (and see also \cite{Abbe2017}). We note that in the cases our algorithm is guaranteed to solve the exact recovery problem (see Theorem~\ref{thm:MainSuccess}), we assume that $\frac{1}{k}\left(\sqrt{P} - \sqrt{Q}\right)  \to \infty$, well above this theoretical bound.
\end{remark} 
Given the bound \eqref{eq:DetectionBound}, the challenge then is to construct efficient algorithms to solve the Exact Recovery Problem. There are myriad algorithmic approaches, such as the spectral approach (originally proposed by Fiedler in \cite{Fiedler1975} for the two cluster case), hierarchical approaches like the DIANA algorithm popular in bioinformatics  (\cite{Kaufman1990}) and message passing algorithms (\cite{Frey2007}), to name a few. Recently, two new classes of algorithms, namely degree-profiling (\cite{Abbe2015}) and semidefinite progamming approaches (\cite{Abbe2016}, \cite{Montanari2015}, \cite{Hajek2016} and \cite{Le2015} among others) have been shown to solve the exact recovery problem with high probability right down to the theoretical bound \eqref{eq:DetectionBound}.  Degree-profiling even runs in quasi-linear time, although it requires the parameters $p$, $q$ and $k$ as inputs, making it less than ideal for analysing real world data sets. As the new algorithm we propose is most closely related to the spectral approach, let us recall this algorithm here (as formulated by Ng, Jordan and Weiss in \cite{Ng2001}). 

\begin{algorithm}[H]
\caption{The Spectral Clustering Algorithm (SC)}
\label{algorithm:SC}
{\bf Input:} the adjacency matrix $\bfA$.
\begin{algorithmic}
\State (1) Form the degree matrix $\mathbf{D} = \text{diag}\left(d_1,\ldots, d_n\right)$ and Laplacian $\bfL := \mathbf{D}^{-1/2}\bfA\mathbf{D}^{-1/2}$
\State (2) Find the $k$ (orthogonal) eigenvectors $\bfv_1,\ldots, \bfv_k$ corresponding to the $k$ largest eigenvalues of $\bfL$.
\State (3) Form the matrix $\mathbf{U} = [\bfv_1,\ldots, \bfv_k]$ and normalize the \emph{rows} to get $\mathbf{V}$
\State (4) Sort the rows of $\mathbf{V}$ into $k$ clusters $B_1,\ldots, B_k$ using $k$-means.
\State (5) Assign vertex $i$ to community $C_j$ if and only if row $i$ is in $B_j$.
\end{algorithmic}
{\bf Output:} Communities $C_1,\ldots, C_k$.
\end{algorithm}
 
We mention that notions from Compressive Sensing have been applied to community detection before,  notably in the semidefinite programming approaches mentioned above, and in \cite{Tremblay2016} where signal processing techniques for functions defined on a graph $G$ are used to speed up the computation of the eigenvectors of $\bfL$. Our approach is distinct from these.  To the best of the authors' knowledge, the study in this paper is the first attempt to find the indicator vectors $\bfone_{C_i}$ directly using sparse recovery.    

\section{Preliminary on Graph Theory}
\label{section:GraphTheory}
\subsection{Some Elementary Notions and Definitions}
Formally, by a graph $G$ we mean a set of vertices $V$ together with a 
subset $E\subset \{\{u,v\}:\ u,v \in V\}$ of edges\footnote{We only consider undirected graphs}. As we are only 
concerned with finite graphs, we shall always identify the vertex set 
$V$ with a finite set of consecutive  
natural numbers: $V = [n] := \{1,2,\ldots, n\}$.  The \emph{degree} of any vertex $i\in G$ is the total 
number of edges incident to $i$, that is  $d_{i} = |\{\{i,j\}\in E\}| $. 

A \emph{subgraph} $G^{'}$ of $G$ is a subset of vertices 
$V^{'}\subset V$ together with a subset of vertices \mbox{$E^{'} \subset 
E\cap V^{'}\times V^{'}$}. Given any subset $S\subset V$, we denote by 
$G_{S}$ the subgraph with vertex set $S$ and edge set \emph{all} edges 
$\{i,j\}$ with $i,j\in S$. A \emph{path} in $G$ is a set of `linked' edges 
$\{\{i_{1},i_{2}\},\{i_{2},i_{3}\},\ldots, \{i_{k-1},i_{k}\}$, and we 
say that $G$ is \emph{connected} if there is a path linking any two 
vertices $i,j\in V$, and 
\emph{disconnected otherwise}. If $G$ is disconnected, any subgraph 
$G_{S}\subset G$ which 
is connected and maximal with respect to the property of being 
connected is called a \emph{connected 
component}. 

If $G$ is connected, we define the \emph{diameter} of $G$ 
to be the length of (i.e. the number of edges in) the longest path.  
Given any $i\in V$ and any non-negative integer $j$, we define the ball 
$B_{j}(i)\subset V$ to be the set of all vertices connected to $i$ be a 
path of length $j$ or shorter. A good reference on elementary graph theory is \cite{Aldous2000a}.

\subsection{The Graph Laplacian}
To any graph $G$ with $|V|=n$ we associate a symmetric, $n\times n$, 
non-negative matrix called the \emph{adjacency matrix} $\bfA$, defined as $\bfA_{ij} = \bfA_{ji} = 1$ if $\{i,j\}\in E$ and $\bfA_{ij} = 0$ otherwise. The graph Laplacians of $G$ are defined as follows. 
\begin{definition}
Let $\bfA$ denote the adjacency matrix of a graph $G$ and let $\mathbf{D}$ denote 
the matrix $\text{diag}(d_{1},\ldots, d_{n})$ where $d_{i}$ is the 
degree of the $i$-th vertex. We define the \emph{normalized, symmetric 
graph Laplacian} of $G$ as $\bfL_{s} := \mathbf{I}- \mathbf{D}^{-1/2}\bfA\mathbf{D}^{-1/2}$ and 
the \emph{normalized, random walk graph Laplacian} as 
$\bfL_{rw} := \mathbf{I} - \mathbf{D}^{-1}\bfA$.
\label{Laplace_defn}
\end{definition}

We first have a few basic properties of graph Laplacians. 
\begin{theorem} Suppose that $\bfL = \bfL_{s}$ or $\bfL_{rw}$. 
We have the following properties:
\label{thm:SpectralProp}
\begin{enumerate}
\item The eigenvalues of $\bfL$ are real and non-negative.
\item $\lambda_{n-1} \leq 2$
\item Let $\lambda_{1}\leq \lambda_{2} \leq \ldots \leq \lambda_{n}$ 
denote the eigenvalues of $L$ in ascending order. Let $k$ denote the 
number of connected components of $G$. Then $\lambda_{i} = 0$ for 
$i\leq k$ and $\lambda_{i} > 0$ for $i> k$. 
\end{enumerate}
\end{theorem}
\begin{proof}
For $\bfL = \bfL_{s}$, Items 
$1$ to $3$ follow from Lemma 1.7 in \cite{Chung1999}. If $\bfL = 
\bfL_{rw}$, observe that $\bfL_{rw} = \mathbf{D}^{-1/2}\bfL_{s}\mathbf{D}^{1/2}$, 
and so the eigenvalues of $\bfL_{rw}$ and $\bfL_{s}$ coincide. 
Hence the above hold for $\bfL_{rw}$ as well.
\end{proof}

For any subset $S \subset [n]$ we define its \emph{indicator vector}, denoted $\mathbf{1}_{S}\in\mathbb{R}^{n}$, by $(\mathbf{1}_{S})_{i} = 1$ if $i \in S$ and  zero otherwise. Let $C_1, \cdots, C_k$ be the clusters of $G$. One can check (and see also \cite{Luxburg2007} proposition 2) that $\bfL_{rw} \bfone_{C_i}=0$ for $i=1, \cdots, 
k$. Thus we have the following:
\begin{theorem}
\label{thm:IndicatorVecsKernel}
 The indicator vectors of the connected components of $G$, 
$\bfone_{C_1},\ldots, \bfone_{C_k}$, form a basis for the zero 
eigenspace (i.e. the kernel) of $\bfL_{rw}$.
\end{theorem}

For the rest of this paper, by $\bfL$ we shall mean $\bfL_{rw}$. We shall refer to the $i$-th column of $\bfL$ as 
$\ell_{i}$. One can easily check that:
\begin{equation*}
(\ell_{i})_{k} = \delta_{ik} - \frac{\bfA_{ik}}{d_k}.
\end{equation*}
Finally, we shall denote by $\bfL_{-i}$ the submatrix of $\bfL$ obtained by dropping the column $\ell_i$.

\subsection{Random Graph Theory}
\label{sec:RandomGraphTheory}
As outlined in \S 2, the Stochastic Block Model (SBM) is a widely used mathematical model 
of a random graph with clusters. 
\begin{definition}
Given $n = kn_0$, fix a partition of $V := [n]$ 
into $k$ subsets $C_1,\ldots, C_k$ of equal size $n_0$.  We say $G$ is drawn from the SBM $ 
\mathcal{G}(n,k,p,q)$ if, for all $i,j \in [n]$ with $i\neq j$ , the 
edge $\{i,j\}$ is inserted independently and with probability 
$p \text{ if } i,j \in C_{a}$ for some $a=1, \cdots, k$ and 
$q$ otherwise.
\end{definition}
As we area interested in clustering we assume that $q << p$. We emphasize that the partition $V = C_{1}\cup C_{2} \cup \ldots \cup 
C_{k}$ is fixed before any edges are assigned. We note that the subgraphs $G_{C_a}$ are i.i.d instances of a simpler 
random graph model, the Erd\"{o}s-R\'{e}nyi (ER) model 
$\mathcal{G}(n_0,p)$, first introduced in \cite{Erdos59}
\begin{definition}
We say $H$ is drawn from the ER model 
$\mathcal{G}(n_0,p)$ if $H$ has $n_0$ vertices and for all $i,j\in 
[n_0]$ the edge $\{i,j\}$ is inserted independently and with 
probability $p$.
\end{definition}
Returning to the SBM, for any vertex $i$ in 
community $C_a$, we define its \emph{in-community degree} as $d^{0}_{i} 
= \#\{\{i,j\}\in E: j \in C_{a}\}$ and its \emph{out-of-community 
degree} as $d_{i}^{\epsilon} := \#\{\{i,j\}\in E: j \notin C_{a}\}$. 
One can easily see that 
\begin{equation*}
\mathbb{E}[d^{0}_{i}] = p(n_0-1) \text{ and } 
\mathbb{E}[d_{i}^{\epsilon}] = q(n-n_0)
\end{equation*} 
and by definition $d_{i} = d^{0}_{i} + d^{\epsilon}_{i}$. In fact 
$d_{i}^{0}$ is the degree of $i$ considered as a vertex in the ER subgraph 
$G_{C_a}$. An important fact about degrees in ER random graphs is that 
they concentrate around their mean:
\begin{theorem}
Suppose $G$ is drawn from  $\mathcal{G}(n,k,p,q)$. 
\begin{enumerate}
\item For any $\alpha > 0$, if $p \geq \dfrac{4\ln(n)}{\alpha^2n_0}$ then with probability at least $1 - 1/n$:
\begin{equation}
(1-\alpha)n_0p \leq d^{0}_{i} \leq (1+\alpha)n_0p \quad \text{ for all } i\in \{1, 2, \cdots, n\}
\label{eq:AlphaDegBound}
\end{equation}
\item In particular, if $p \geq \dfrac{4k(\ln(n))^{2}}{n}$ then \eqref{eq:AlphaDegBound} holds, with probability at least $1-1/n$, for $\alpha = 1/\sqrt{\ln(n)} = o(1)$.
\end{enumerate}
\label{lemma:AlphaDegBound}
\end{theorem}

\begin{proof}
This theorem is a variation on a well known result for ER graphs (e.g. theorem 3.6 in \cite{Frieze2016}). Each $d^{0}_{i}$ follows the binomial distribution with parameters $n_0-1$ and $p$, so by the Chernoff bound: $\mathbb{P}\left[|d^{0}_{i} - n_0p|\geq \alpha n_0p\right] \leq e^{-\alpha^2n_0p/2}$. Hence:
\begin{align*}
\mathbb{P}\left[\max_{i\in [n]}|d^{0}_{i} - n_0p|\geq \alpha n_0p\right] & = \mathbb{P}\left[|d^{0}_{1} - n_0p|\geq \alpha n_0 p \text{ or } |d^{0}_{2} - n_0p|\geq \alpha n_0 p \text{ or } \ldots \text{ or } |d^{0}_{n} - n_0p|\geq \alpha n_0 p \right] \\
&\leq \sum_{i=1}^{n} \mathbb{P}\left[ |d^{0}_{i} - n_0p| \geq \alpha n_0p \right] = ne^{-\alpha^2n_0p/2}
\end{align*}
Thus $\displaystyle \mathbb{P}\left[\max_{i\in [n_0]}|d^{0}_{i} - n_0p|\leq \alpha n_0p\right] = 1 - ne^{-\alpha^2n_0p/2}$ If $p \geq \dfrac{4\ln(n)}{\alpha^2n_0} = \dfrac{2\ln(n^{2})}{\alpha^2n_0}$ then:
\begin{equation*}
1 - ne^{-\alpha^2n_0p/2} \leq 1 - ne^{-\ln(n^2)} = 1 - 1/n
\end{equation*}
This proves part $1$. Part $2$ follows by taking $\alpha = 1/\sqrt{\ln(n)}$, in which case the lower bound on $p$ becomes:
\begin{equation*}
\frac{4\ln(n)}{\alpha^2n_0} = \frac{4(\ln(n))^{2}}{n_0} = \frac{4(\ln(n))^{2}}{n/k} =  \frac{4k(\ln(n))^{2}}{n}
\end{equation*}
\end{proof}

The second remarkable property of the ER model is that the eigenvalues of $\bfL$ also concentrate around their mean:
\begin{theorem}
\label{theorem:ChungLuVu}
Let $\bfL$ be the Laplacian of a random graph drawn from $\mathcal{G}(n_0,p)$ with $p >> (\ln(n_0))^2/n_0$ and let $\lambda_1 \leq \lambda_2 \leq \ldots \leq \lambda_{n_0}$ denote its eigenvalues. Then almost surely  \footnote{Given a family of random graph models $\mathcal{G}_{n}$ we say that some graph property $P$ holds almost surely if \mbox{$\mathbb{P}[G\in  \mathcal{G}_{n}, G \text{ does not have P}] = o(1)$} with respect to $n$}:
\begin{equation*}
\max_{i\neq 1} |1 - \lambda_{i}| \leq \left(1 + o(1)\right)\frac{4}{\sqrt{pn_0}}  + \frac{g(n_0)\log^{2}(n_0)}{pn_0}
\end{equation*}
where $g(n_0)$ is a function tending to infinity arbitrarily slowly.
\end{theorem}
\begin{proof}
Given that the expected degree of each vertex in $G$ is $pn_0$, 
this is just Theorem  3.6 in \cite{Chung2003}
\end{proof}

\begin{remark}
\label{remark:LambdaMin}
For our purposes, it will be enough to note that this gives:
\begin{equation*}
\lambda_{2} \geq 1 - \frac{4/\sqrt{p}}{\sqrt{n_0}} - o(\frac{1}{\sqrt
{n_0}}) \text{ and } \lambda_{n_0} \leq 1 + 
\frac{4/\sqrt{p}}{\sqrt{n_0}} + o(\frac{1}{\sqrt{n_0}})
\end{equation*}
almost surely. 
\end{remark}

\section{Preliminaries on Compressive Sensing}
\label{section:SparseRecovery}
Let $\Phi\in \mathbb{R}^{m\times N}$ and $\bfy\in \mathbb{R}^{m}$ with 
$m < N$. We say that a vector $\bfx$ is \emph{sparse} if it has few 
non-zero entries relative to its length. We follow the convention of 
defining the `$0$ quasi-norm' as:
\begin{equation*}
\|\bfx\|_{0} = \#\{x_{i}: \ x_{i} \neq 0 \}
\end{equation*}
and we say $\bfx$ is $s$-sparse if $\|\bfx\|_{0} \leq s$. Compressive sensing is concerned with solving
\begin{equation}
\label{eq:SparseRecovery}
\argmin\{ \|\bfx\|_0: \quad \bfx\in \mathbb{R}^N, \ \Phi \bfx= \bfy\}
\end{equation}
in the case where $m < N$ (that is, when the linear system is underdetermined). We call \eqref{eq:SparseRecovery} the \emph{Sparse Recovery Problem}. One also considers the \emph{Perturbed Sparse Recovery Problem} where $\bfy  = \Phi\bfx^{*} + \bfe$ with $\|\bfe\|_{2} \leq \eta << \|\bfy\|_{2}$ and we wish to solve:
\begin{equation}
\label{eq:SparseRecoveryPerturbedFirst}
\bfx^{\#} :=  \argmin\{ \|\bfx\|_0: \quad \bfx\in \mathbb{R}^N, \ \|\Phi \bfx - \bfy\|_{2} \leq \eta \}
\end{equation}
while guaranteeing that $\bfx^{\#}$ is a good approximation to $\bfx^{*}$ by bounding $\|\bfx^{\#} - \bfx^{*}\|_{2}$? We remark that problem \eqref{eq:SparseRecoveryPerturbedFirst} is equivalent to the dual problem:
\begin{equation}
\label{eq:SparseRecoveryPerturbed}
\bfx^{\#} :=  \argmin\{ \|\Phi \bfx - \bfy\|_{2}: \quad \bfx\in \mathbb{R}^N, \  \|\bfx\|_0 \leq n_0 \}
\end{equation}
We refer the reader to \cite{Foucart2013} for an excellent introduction to the area. 

\subsection{Computational Algorithms}
Many numerical algorithms have been invented for solving \eqref{eq:SparseRecovery} and \eqref{eq:SparseRecoveryPerturbed}, for example, $\ell_1$ convex minimization and its variations, hard thresholding iteration and its variations, greedy approaches such as orthogonal matching pursuit (OMP) 
as well as more exotic approaches like $\ell_q$ nonconvex 
minimization. See, for example, \cite{CT05}, 
\cite{CWB08}, \cite{BT09},  
\cite{Tropp2004}, \cite{BD09}, \cite{F11}, \cite{FL09}.
Due to its efficiency, 
we shall focus on the greedy approach in this paper, specifically the 
Orthogonal Matching Pursuit (OMP) and Subspace Pursuit (SP) algorithms, see Algorithms \ref{algorithm:OMP} and \ref{algorithm:SP}, respectively. For notational convenience, we shall follow \cite{Foucart2013} and denote by $\Phi_{S}$ the column submatrix of $\Phi$ consisting of the columns indexed by the subset $S \subset [N]$. For a vector $\mathbf{x} \in \mathbb{R}^{N}$ we denote by $\mathbf{x}_{S}$ either the subvector in $\mathbb{R}^{|S|}$ consisting of the entries $x_i$ for $i$ indexed by $S$, or the vector 
\begin{equation*}
(\mathbf{x}_{S})_{i} = \left\{ \begin{array}{cc} x_i & \text{ if } i \in S \\ 0 & \text{ if } i \notin S \end{array}\right.
\end{equation*}
It should always be clear from the context which definition we are referring to. We also remind the reader of the following operations on vectors $\bfv\in\mathbb{R}^{n}$, defined in \cite{Foucart2013}:
\begin{align*}
& \mathcal{L}_{s}(\bfv) = \text{ index set of $s$ largest absolute entries of } \bfv \hbox{ and }
\mathcal{H}_{s}(\bfv) = \bfv_{\mathcal{L}_{s}(\bfv)} = \left\{ \begin{array}{cc} v_i & \text{ if } i \in S \\ 0 & \text{ if } i \notin S \end{array}\right.
\end{align*}
where $\mathcal{H}_{s}$ is sometimes referred to as the \emph{Hard Thresholding Operator}.

\begin{algorithm}[H]
\caption{The OMP Algorithm}
\label{algorithm:OMP}
Inputs: $\bfy$ and $\Phi$
\begin{algorithmic}
\State Initialize: $\bfr^{(0)} = \bfy$, $S^{(0)} = \emptyset$, 
$\bfx^{(0)} = \mathbf{0}$ and $\alpha = 0$.
\While{True}
	\State $\alpha \gets \alpha +1$
	\State $i^{(\alpha)} := \mathcal{L}_{1}(\Phi^\top \bfr^{\alpha-1})$
	\State $S^{(\alpha)} = S^{(\alpha-1)}\cup\{i^{(\alpha)}\}$
	\State $\bfx^{(\alpha)} := \argmin 
\{||\Phi_{S^{(\alpha)}}\mathbf{z}_{S^{(\alpha)}} - \bfy||_{2}: \ 
\mathbf{z}\in\mathbb{R}^{N} \text{ and } \supp(\mathbf{z})\subset 
S^{(\alpha)}\}$
	\State $\bfr^{(\alpha)} := \bfy - \Phi\bfx^{(\alpha)}$
	\If{Stopping Criterion met}
		\State Output $\bfx^{\#} := \bfx^{\alpha}$
		\State Break
	\EndIf
\EndWhile
\end{algorithmic}
\end{algorithm} 

\begin{algorithm}[H]
\caption{The SP algorithm (\cite{Dai2009})}
\label{algorithm:SP}
Inputs: $\bfy$, $\Phi$ and an integer $s\ge 1$
\begin{algorithmic}
\State {\bf Initialization}:
	\Indent
	\State (1) $T^{0} = \mathcal{L}_{s}(\Phi^\top \bfy)$.
	\State (2) $\bfx^{0} = \argmin\{ \|\bfy - \Phi_{T^{0}}\bfx\|_{2}: \ \supp(\bfx) \subset T^{0}\}$
	\State (3) $\bfr^{0} = \bfy - \Phi_{T^{0}}\bfx^{0}$
	\EndIndent
\State {\bf Iteration}:
	\Indent
	\For{$k = 1:k_{\max}$}  
	\State (1) $\hat{T} ^{k} = T^{k-1}\cup 	\mathcal{L}_{s}\left(\Phi^\top \bfr^{k-1}\right)$
	\State (2) $\bfu = \argmin\{  \|\bfy - \Phi_{\hat{T}^{k}}\bfx\|_{2} : \ \mathbf{x}\in\mathbb{R}^{N} \text{ and } \supp(\bfx) \subset \hat{T}^k\}$
	\State (3) $T^{k} = 	\mathcal{L}_{s}(\bfu)$ and $\bfx^{k} = \mathcal{H}_{s}(\bfu)$
	\State (4) $\bfr^{k} = \bfy - \Phi_{T^{k}}\bfx^{k}$
	\EndFor
	\EndIndent
\end{algorithmic}
\end{algorithm}

The most common stopping criteria for Algorithm~\ref{algorithm:OMP} are $\alpha = k$, or 
$\|\bfr^{(\alpha)}\|_{2} < \epsilon$ for a given $k$ or $\epsilon$.  
A sufficient condition to guarantee the convergence of  
Algorithm~\ref{algorithm:OMP} is the following

\begin{theorem}
For any $S\subset [N]$, if  $\Phi_S$ is injective and 
satisfies
\begin{equation}
 \| \Phi_{S}^{\dagger}\Phi_{S^{c}}\|_{1\rightarrow 1} < 1,  
\label{Tropp} 
\end{equation}
where $\Phi^{\dagger}_{S}$ is the pseudo-inverse of $\Phi_{S}$, then any vector $\mathbf{x}$ with support $S$ is recovered in 
at most $s = |S|$ steps of OMP.
\label{thm:ExactRecovCond}
\end{theorem}
We refer to \cite{Tropp2004} for a proof. See also Remark 3.6 in 
\cite{Foucart2013}.

\begin{remark}
In \cite{Dai2009} where the SP algorithm is introduced, they 
suggest solving the least squares problems (that is $(2)$ in 
Initialization and $(2)$ in Iteration) exactly. In our 
implementation, we use {\tt MATLAB}'s {\tt lsqr } algorithm 
to solve them approximately, to a high precision. As pointed 
out by \cite{Needell2009} in their analysis of CoSaMP, a 
very similar algorithm, this does not affect the convergence 
analysis of the algorithm.
\end{remark}
	
\subsection{Fundamental Concepts}	 
A very important concept, the \emph{restricted isometry property} (RIP), 
introduced by Cand\'{e}s and Tao (\cite{CT05}) plays a critical 
role in the study of the existence and uniqueness of a sparse solution 
from a sensing matrix and whether the solution of an $\ell_1$ minimization 
is the sparse solution. Another concept, \emph{mutual coherence}, 
introduced by Donoho and his collaborators in (\cite{Donoho2003} is also used in this study.

\begin{definition}
\label{RIPdef}
Letting $0<s<m < N$ be an integer and  $\Phi_S$ be a submatrix of $\Phi$ 
which consists of columns of $\Phi$ whose column indices are in  
$S\subset [N]$. The \emph{restricted isometry constant} 
(RIC) $\delta_s$ of $\Phi$ is the smallest quantity such that
\begin{equation}
\label{RIP}
(1-\delta_s) \|{\bf x}\|_2^2 \le \|\Phi_S {\bf x}\|_2^2 \le (1+
\delta_s)\|{\bf x}\|_2^2
\end{equation}
for all subsets $S$ with with cardinality $|S|\le s$. If a matrix 
$\Phi$ has such a constant $\delta_s< 1$ for some $s$, $\Phi$ is said 
to possesses RIP of order $s$. It is known that
\begin{equation}
\label{Foucart}
\delta_s = \min_{S\subset \{1, \cdots, n\}\atop |S|\le s} 
\|\Phi_S^\top \Phi_S - I_s\|_{2\to 2},
\end{equation}
where $I_s$ is the identity matrix of size $s\times s$. 
\end{definition}

\begin{definition}
The \emph{coherence} of a matrix $\Phi$, denoted $\mu$, is the largest 
normalized inner product between its columns:
\begin{equation*}
\mu := \max_{i\neq j} \frac{|\langle \phi_{i},\phi_{j}\rangle |}{
\|\phi_i\|_{2}\|\phi_{j}\|_{2}}. 
\end{equation*}
\end{definition}

\subsection{Totally Perturbed Compressive Sensing}
\label{section:PerturbedCS}
Frequently it is useful to modify problem 
\eqref{eq:SparseRecoveryPerturbed} further to allow for small 
perturbations in the observed measurement matrix. That is, suppose that 
$\bfy = \Phi\bfx^{*} + \bfe$ and let $\hat{\Phi} = \Phi + \bfE$, where $\bfE$ 
is a small perturbation matrix. Denoting again:
\begin{equation}
\label{eq:SparseRecoveryTotallyPerturbed}
\bfx^{\#} :=  \argmin\{\|\hat{\Phi} \bfx - \bfy\|_{2}: \quad \bfx\in \mathbb{R}^N, \ \|\bfx\|_0  \leq s \}, 
\end{equation}
can we still guarantee that $\bfx^{\#}$ is a good approximation to 
$\bfx^{*}$ by bounding $\|\bfx^{\#} - \bfx^{*}\|_{2}$? 
Problem \eqref{eq:SparseRecoveryTotallyPerturbed} is called 
the \emph{Totally Perturbed Sparse Recovery Problem}. Analyzing \eqref{eq:SparseRecoveryTotallyPerturbed} is often more important for applications than analyzing \eqref{eq:SparseRecoveryPerturbed}, as frequently we only know the measurement matrix $\Phi$ to within a certain error tolerance.

\begin{theorem}[Herman and Strohmer, \cite{Herman2010}]
\label{thm:PerturbedRIC}
Suppose that $\hat{\Phi} = \Phi + \bfE$. Let $\delta_{s}$ and $\hat{\delta}_{s}$ denote the $s$ restricted isometry constants of $\Phi$ and $\hat{\Phi}$ respectively. Define\footnote{For a matrix $\Phi$ we denote by $\|\Phi\|_{p}$ the induced operator norm $\|\Phi\|_{p} = \max_{\bfx \neq \mathbf{0}}\frac{\|\Phi\bfx\|_{p}}{\|\bfx\|_{p}}$ By $\|\Phi\|_{p,s}$ we mean the semi-norm $\displaystyle\max_{S\subset [n]\atop |S| = s}\|\Phi_{S}\|_{p}$} $\epsilon^{s}_{\Phi} := \|E\|_{2,s} /\|\Phi \|_{2,s}$. Then:
\begin{equation*}
\hat{\delta}_{s} \leq (1+\delta_{s})\left(1+\epsilon^s_{\Phi}\right)^{2} - 1
\end{equation*}
\end{theorem}

To bound the error $\|\bfx^{\#} - \bfx^{*}\|_{2}$, we need the following result.
Define $\bfy = \Phi\bfx $  and suppose that $\hat{\bfy} = \Phi\bfx +\bfe$ is the observed (perturbed) measurement vector. Suppose further that we only have access to $\hat{\Phi} = 
\Phi + \bfE$, a small perturbation of the measurement matrix.  
Define $\epsilon_{\bfy} := \|\bfe\|_{2}/\|\bfy\|_{2}$, $\epsilon_{\Phi} = \|\bfE\|_{2}/\|\Phi\|_{2}$ and let $\epsilon^{s}_{\Phi}$ be as above. 

\begin{theorem}
\label{thm:PerturbedSP}
Suppose that $\bfx$ is $s$-sparse. Define the following constants:
\begin{equation*}
\hat{\rho} = \frac{\sqrt{2\hat{\delta}_{3s}^{2}(1+\hat{\delta}_{3s}^{2})}}{1 - \hat{\delta}_{3s}^{2}} \quad \text{ and } \quad  \hat{\tau} = \frac{(\sqrt{2} + 2)\hat{\delta}_{3s}}{\sqrt{1 - \hat{\delta}_{3s}^{2}}}(1 - \hat{\delta}_{3s})(1 - \hat{\rho}) + \frac{2\sqrt{2}+1}{(1 - \hat{\delta}_{3s})(1-\hat{\rho})}. 
\end{equation*}
Suppose that $\hat{\delta}_{3s} \leq 0.4859$. Then after $m = \ln(\epsilon_{\Phi} + \epsilon_{\bfy})/\ln(\hat{\rho})$ iterations  of Subspace Pursuit (Algorithm \ref{algorithm:SP}) applied to problem \eqref{eq:SparseRecoveryTotallyPerturbed}, we have
\begin{equation*}
\frac{\|\bfx - \bfx^{m}\|_{2}}{\|\bfx\|_{2}} \leq \left(\hat{\tau}\frac{\sqrt{1 + \hat{\delta}_{s}}}{1 - \epsilon^{s}_{\Phi}} + 1\right)(\epsilon^{s}_{\Phi} + \epsilon_{\bfy}).
\end{equation*}
\end{theorem}

\begin{remark}
This is Theorem 2 in \cite{Li2016}, adapted to the case the $\bfx$ is $s$-sparse (the result \emph{in. loc. sit.} is for the more general case where $\bfx$ is compressible).
\end{remark} 

\section{The RIP of Laplacian of Random Graphs}
\label{section:RIPLaplacian}
We first study the RIP for the Laplacian of a connected ER graph drawn from $\mathcal{G}(n_0,p)$. 
We then extend this to a result on the RIP for graphs drawn from the SBM $\mathcal{G}(n,k,p,0)$, as 
these can be thought of as a disjoint union of $k$ ER graphs. Finally, we extend to graphs drawn from 
$\mathcal{G}(n,k,p,q)$ for $0< q << p$ using a perturbation argument.

\subsection{RIP for Laplacian of $\mathcal{G}(n_0,p)$}
\begin{lemma}
\label{lemma:Laplacian1Comp}
Let $\bfL$ be the Laplacian of a connected graph $G$ with $n_0$ vertices. Let $S\subset [n_0]$ 
with $|S| = s < n_0$. 
Then  $\sigma_{\min}(\bfL_{S}) \geq (1 - s/n_0)\lambda_{2}$ and $\delta_{s}  = \max\{1 - \left(1 - 
s/n_0\right)\lambda_{2}^{2} , \lambda_{n_0}^{2} - 1\}$, where $\lambda_{i}$ denotes the 
$i$-th eigenvalue of $\bfL$, ordered from smallest to largest.
\end{lemma}
\begin{proof}
Let $\bfw_{1},\ldots, \bfw_{n_0}$ be an orthonormal basis of eigenvectors of $\bfL$ with eigenvalues $\lambda_1\le \cdots\le 
\lambda_{n_0}$, where $\lambda_1=0$ and $\lambda_2>0$.   For any $\bfv$ supported on $S$ with $\|\bfv\|_{2} = 1$, write $\bfv = \sum_{i=1}^{n_0}\alpha_{i}\bfw_{i}$. Note that $\bfw_{1} = \frac{1}{\sqrt{n_0}} \bfone$ and 
$\ker(L) = \text{span}(\bfw_{1})$.  Then:
\begin{equation*}
\|\bfL\bfv\|_{2}^{2} = \left\|\bfL\left(\sum_{i=1}^{n}\alpha_{i}\bfw_{i}\right)\right\|_{2}^{2} = \left\|\sum_{i=2}^{n}
\alpha_{i}\lambda_{i}\bfw_{i}\right\|_{2}^{2} = \sum_{i=2}^{n} \alpha_{i}^{2}\lambda_{i}^{2}  \geq \left(\sum_{i=2}^{n} 
\alpha_{i}^{2}\right)\lambda_{2}^{2}
\end{equation*}
Because $\sum_{i=1}^{n_0}\alpha_{i}^{2} = \|\bfv\|_{2}^{2} = 1$ we have that $\sum_{i=2}^{n_0}\alpha_{i}^{2} = 1 - 
\alpha_{1}^{2}$. Thus $\|\bfL\bfv\|_{2}^{2} \geq \lambda_{2}^{2}(1 - \alpha_{1}^{2})$ and so clearly this quantity is 
minimized by making $\alpha_{1}$ as large as possible. Observe that:
\begin{equation*}
\alpha_{1} = \left(\frac{1}{\sqrt{n_0}}\bfone\right) \cdot \bfv \leq \frac{1}{\sqrt{n_0}}\|\bfv\|_{1} \leq 
\frac{1}{\sqrt{n_0}}(\sqrt{s}\|\bfv\|_{2}) = \frac{\sqrt{s}}{\sqrt{n_0}}.
\end{equation*}
We remark that this bound on $\alpha_{1}$ is sharp, and is achieved by taking $\bfv = \frac{1}{\sqrt{s}}\bfone_{S}$. 
Hence:
\begin{align*}
\sigma_{\min}(\bfL_{S}) =  \min_{\supp(\bfv)\subset S \atop \|\bfv\|_{2} = 1}\|L\bfv\|_{2}^{2}  \geq \left(1 - 
\left(\sqrt{\frac{s}{n_0}}\right)^{2}\right)\lambda_{2}^{2} = \left(1 - \frac{s}{n_0}\right)\lambda_{2}^{2}.
\end{align*}
On the other hand:
\begin{equation*}
\max_{\supp(\bfv)\subset S \atop \|\bfv\|_{2} = 1}\|\bfL\bfv\|_{2}^{2} \leq \max_{\|\bfv\|_{2} = 1}
\|\bfL\bfv\|_{2}^{2} = \lambda_{n_0}^{2}.
\end{equation*}
Hence for any $S\subset [n]$ with $|S| = s$, and any $\bfv$ with $\supp(\bfv)\subset S$ and $\|\bfv\|_{2} = 1$, we 
have:
\begin{equation*}
\left(1 - \frac{s}{n_0}\right)\lambda_{2}^{2} \leq \|L\bfv\|_{2}^{2} \leq \lambda_{n_0}^{2}
\end{equation*}
and the claim about $\delta_s$ follows.
\end{proof}

\begin{theorem}[RIP for Laplacian of ER graphs]
\label{thm:RIPER}
Suppose that $G \in \mathcal{G}(n_0,p)$ with Laplacian $\bfL$ and suppose that $p >> (\ln(n_0))^2/n_0$. If 
$s = \gamma n_0$ with $\gamma \in (0,1)$ then 
\begin{equation*}
\delta_{s} \leq \gamma +(1-\gamma)\frac{8p^{-1/2}}{\sqrt{n_0}} + o(\frac{1}{\sqrt{n_0}})
\end{equation*}
almost surely \footnote{Recall that here and throughout this paper we say that a property holds almost surely if the probability that \emph{it does not hold} is $o(1)$ with respect to $n_0$ (or sometimes $n$)}
\end{theorem}
\begin{proof}
By Theorem~\ref{theorem:ChungLuVu}, and Remark~\ref{remark:LambdaMin}, 
\begin{equation*}
\lambda_{2} \geq 1 - \frac{4p^{-1/2}}{\sqrt{n_0}} - o(\frac{1}{\sqrt{n_0}}) \text{ and } \lambda_{n_0} \leq 1 + \frac{4p^{-1/2}}{\sqrt{n_0}} + o(\frac{1}{\sqrt{n_0}})
\end{equation*}
almost surely. Combining this with Lemma~\ref{lemma:Laplacian1Comp}, we get that:
\begin{align*}
\delta_{s} & \leq 1 - (1-s/n_0)\left(1 - \frac{4p^{-1/2}}{\sqrt{n_0}} - o(\frac{1}{\sqrt{n_0}}) \right)^2 \\
	& = 1 - (1-\gamma)\left(1 - \frac{8p^{-1/2}}{\sqrt{n_0}} + o(\frac{1}{\sqrt{n_0}})\right) \\
    & = 1 - \left( 1 - \gamma -(1-\gamma)\frac{8p^{-1/2}}{\sqrt{n_0}}  + o(\frac{1}{\sqrt{n_0}})\right) 
= \gamma +(1-\gamma)\frac{8p^{-1/2}}{\sqrt{n_0}} + o(\frac{1}{\sqrt{n_0}}). 
\end{align*}
as claimed.
\end{proof}

We conclude this subsection with a lower bound for $\|\bfL\|_{2,s}$ (see \S \ref{section:PerturbedCS} for the definition of $\|\cdot \|_{2,s}$), 
where $\bfL$ is the Laplacian 
of a connected graph.
\begin{lemma}
\label{lemma:LowerBoundSemiNorm}
If $\bfL$ is the Laplacian of a connected graph $G$ with $n_0$ vertices, then 
$\|\bfL\|_{2,s} \geq \lambda_{s-1} $,.
\end{lemma}
\begin{proof}
Recall that:
\begin{equation}
\label{newnorm}
\|L\|_{2,s} := \max_{S\subset[n] \atop |S| = s}\|L_{S}\|_{2} = \max_{S\subset[n] \atop |S| = s}\sigma_{\max}(L_S)
\end{equation}
We shall apply Theorem~\ref{thm:interlacingEigs} below. 
For a matrix such as $\bfL$ which is conjugate to a symmetric nonnegative definite 
matrix, the eigenvalues coincide with the singular values.
Translating the notation of this Theorem~\ref{thm:interlacingEigs}  
into the current situation, for $p = n$, $q = s$, 
$\beta_{i} = \sigma_{i}(\bfL_{S})$ and $\alpha_{i} = \lambda_{i}(\bfL)$. 
Clearly, $\min(m,n) = \min(n,n) = n$ and 
$\min(p,q) = \min(n,s) = s$.  We use 
Theorem~\ref{thm:interlacingEigs} again to get:
\begin{equation*}
\sigma_{\max}(\bfL_{S}) \geq \lambda_{s-1}. 
\end{equation*}
And so $\displaystyle\max_{S\subset[1, \cdots, n] \atop |S| = s}\sigma_{\max}(\mathbf{L}_S) \geq \lambda_{s-1}$. 
\end{proof}

In the proof above, we have used 
the following classic interpolation Theorem  for singular values:  
\begin{theorem}
\label{thm:interlacingEigs}
Let $\bfA$ be an $m\times n$ matrix with singular values $\alpha_1 \leq \alpha_{2} \leq \ldots \leq \alpha_{\min(m,n)}$. Suppose that $\mathbf{B}$ is a $p\times q$ submatrix of $\bfA$, with singular values $\beta_{1}\leq \beta_{2} \leq \ldots \leq \beta_{\min(p.q)}$ Then:
\begin{align*}
 \alpha_{\min(m,n) - i} \geq \beta_{\min(p,q) - i} & \text{ for  } i = 0,1,\ldots, \min(p,q) - 1 
 \end{align*}
 and 
 \begin{align*}
 \beta_{i} \geq \alpha_{\min(m,n) - \min(p,q) - (m-p) - (n-q) -1 + i}
\end{align*}
for integer $i$ satisfying  $\min(p,q) + (m-p) + (n-q) +2 - \min(m,n) \leq i \leq \min(m,n)$. 
\end{theorem}
\begin{proof}
This is Theorem  1 in \cite{Thompson1972}. Note that they use the opposite notational convention: $\alpha_{1} \geq \alpha_{2} \geq \ldots \geq \alpha_{\min(m,n)}$ and $\beta_{1}\geq \beta_{2} \geq \ldots \geq \beta_{\min(p.q)}$. 
\end{proof}

\subsection{RIP for Laplacian of Graphs from $\mathcal{G}(n,k,p,0)$}
\begin{lemma}
\label{lemma:LaplacianManyComp}
Suppose that a graph $G$ has $k$ connected components $C_1, \cdots, C_k$, all of size $n_0$ (for example, 
$G\in\mathcal{G}(n,k,p,0)$). Let $G_{C_1}, \ldots, G_{C_k}$ denote the subgraphs on these components and let 
$\bfL^{i}$ denote their Laplacians. Then for any $s < n_0$, $\delta_{s}(\bfL_{G}) = \max_{i}\delta_{s}(\bfL^{i})$.
\end{lemma}
\begin{proof}
Suppose $S\subset C_{i}$ for some $i$. For simplicity we assume $i=1$, but the other cases are identical. 
In this case $\bfL_{S} = \left[\begin{matrix} \bfL^{1}_{S} \\ \mathbf{0} \end{matrix} \right]$ where $\bfL^{1}$ denotes the 
Laplacian of $G_{C_1}$ and $\mathbf{0}$ here is the zero matrix of the appropriate size. If $\supp(\bfv)\subset S$, 
then $\|\bfL_{S}\bfv\|_{2}^{2} = \|\bfL^{1}_{S}\bfv_{s}\|_{2}^{2}$ and so:
\begin{equation}
(1 - \delta_{s}(\bfL^{1}))\|\bfv\|_{2}^{2}  \leq \|\bfL_{S}\bfv\|_{2}^{2} \leq (1 + \delta_{s}(\bfL^{1})) 
\|\bfv\|_{2}^{2}
\label{eq:ScontainedC} 
 \end{equation}
 It follows that, for all index sets $S$ contained in a single component (i.e. $S\subset C_{i}$ for some $i$), 
we have:
 \begin{equation*}
(1 - \max_{i}\left(\delta_{s}(\bfL^{i})\right))\|\bfv\|_{2}^{2}  \leq \|\bfL_{S}\bfv\|_{2}^{2} \leq (1 + \max_{i}
\left(\delta_{s}(\bfL^{i})\right)) \|\bfv\|_{2}^{2}. 
 \end{equation*} 
Now suppose that $S \not\subset C_{i}$.  Write $S = \cup_{i}S_{i}$ where $S_i := S\cap C_{i}$. Given any $\bfv$ with 
$\supp(\bfv) = S$, we can write $\bfv=\sum_{i=1}^{k}\bfv_{i}$ with $\bfv_i = \mathcal{H}_{S_i}(\bfv)$.  Then:
\begin{equation*}
\|\bfL_{S}\bfv\|_{2}^{2}  = \left\|\sum_{i=1}^{k}\bfL_{S_i}\bfv_{i}\right\|_{2}^{2}
\end{equation*}
Crucially, observe that all the terms $\bfL_{S_i}\bfv_{i}$ have disjoint support. Hence:
\begin{align*}
\|\bfL_{S}\bfv\|_{2}^{2} = \left\|\sum_{i}\bfL_{S_i}\bfv_{i}\right\|_{2}^{2} &= 
\sum_{i} \left\|\bfL_{S_i}\bfv_{i}\right\|_{2}^{2} \geq \sum_{i} \left(1 - \delta_{s_{i}}(\bfL^{i})\right)\|\bfv_{i}\|_{2}^{2} \ \text{ by } \eqref{eq:ScontainedC}\\
& \geq \min_{i}\left(1 - \delta_{s_{i}}(\bfL^{i})\right) \sum_{i}\|\bfv_{i}\|_{2}^{2} = \min_{i}\left(1 - \delta_{s_{i}}(\bfL^{i})\right) \|\bfv\|_{2}^{2} \\
& = \left(1 - \max_{i}\delta_{s_i}(\bfL^{i})\right)\|\bfv\|_{2}^{2} \geq \left(1 - \max_{i}\delta_{s}(\bfL^{i})\right)\|\bfv\|_{2}^{2} 
\end{align*}
with the final inequality holding as $s_{i} \leq s$ for all $i$ and $\delta_{t}$ is non-decreasing in $t$. An identical argument yields that:
\begin{equation*}
\|\bfL_{S}\bfv\|_{2}^{2} \leq \left(1+ \max_{i}\delta_{s}(\bfL^{i})\right)\|\bfv\|_{2}^{2}
\end{equation*}
and so we have 
 \begin{equation*}
(1 - \max_{i}\left(\delta_{s}(\bfL^{i})\right))\|\bfv\|_{2}^{2}  \leq \|\bfL_{S}\bfv\|_{2}^{2} 
\leq (1 + \max_{i}\left(\delta_{s}(\bfL^{i})\right)) \|\bfv\|_{2}^{2}. 
 \end{equation*} 
 This completes the proof.
\end{proof} 

\begin{theorem}[RIP for Laplacian of Graphs from Stochastic Block Model]
\label{thm:RIPBoundq0}
Suppose $G \in \mathcal{G}(n,k,p,0)$, with  $n_0 = n/k$ and Laplacian $\bfL$. 
Suppose further that $p >> (\ln(n_0))^2/n_0$ and that $k$ is $O(1)$ with respect to $n$. If $s = \gamma n_0$ with $\gamma \in (0,1)$ then:
\begin{equation}
\delta_{s}(\bfL) \leq  \gamma +(1-\gamma)\frac{8p^{-1/2}}{\sqrt{n_0}} + o(\frac{1}{\sqrt{n_0}})
\end{equation}
almost surely.
\end{theorem}
\begin{proof}
Because $q = 0$, $G$ will have $k$ connected components with probability $1$. Note that each subgraph $G_{i}$ is an i.i.d ER  graph, drawn from $\mathcal{G}(n_0,p)$.  Let $\bfL^{i}$ denote the Laplacian of $G_{C_i}$. 
By Theorem \ref{thm:RIPER}
\begin{equation}
\delta_{s}(\bfL^i) \leq \gamma +(1-\gamma)\frac{8p^{-1/2}}{\sqrt{n_0}} +  o(\frac{1}{\sqrt{n_0}})
\label{eq:OneBlockRIP}
\end{equation}
almost surely. That is, there exists a function $f(n_0)$ going to $0$ 
as $n_0\to \infty$ such that \eqref{eq:OneBlockRIP} holds with 
probability $1 - f(n_0)$. As the $G_{C_i}$ are i.i.d:
\begin{equation}
\max_{1\leq i\leq k} \delta_{s}(\bfL^i) \leq \gamma +(1-\gamma)
\frac{8p^{-1/2}}{\sqrt{n_0}} +  o(\frac{1}{\sqrt{n_0}})
\label{eq:ManyBlockRIP}
\end{equation}
with probability $(1-f(n_0))^{k}$. As long as $k$ is eventually bounded with 
respect to $n_0$ (i.e. $k$ is $O(1)$ with respect to $n$), $(1-f(n_0))^{k} \to 1$ as $n_0\to \infty$. Hence 
\eqref{eq:ManyBlockRIP} holds almost surely.
\end{proof}

Again, we conclude this subsection with a lower bound for $\|\bfL\|_{2,s}$.
\begin{lemma}
Suppose that a graph $G$ has $k$ connected components $C_1, \cdots, C_k$, all of size $n_0$. Let $G_{C_1}, \ldots, 
G_{C_k}$ denote the subgraphs on these components and let $\bfL^{i}$ denote their Laplacians. Then $\|\bfL\|_{2,s} \geq 
\max_{a\in [k]} \lambda_{s-1}(\bfL^{a})$.
\label{lemma:LoweBoundSemiNorm2}
\end{lemma}
\begin{proof}
We leave the proof to the interested reader. 
\end{proof}

\subsection{RIP for Laplacian of Graphs from $\mathcal{G}(n,k,p,q)$ with $q > 0$}
\label{sectionRIP_q_nonzero}
Finally, we study the RIP for Laplacian of graphs from $\mathcal{G}(n,k,p,q)$ with $q > 0$.  Observe that for $G \in \mathcal{G}(n,k,p,q)$ we may write $\bfA = \bfA^{0} + \bfA^{\epsilon}$ where $\bfA^{0}$ contains the intracluster edges and $\bfA^{\epsilon}$ contains the intercluster edges:
\begin{align*}
\bfA^{0}_{ij} &= \left\{ \begin{array}{cc} 1 & \text{ if } \{i,j\}\in 
E \text{ and } i,j \in C_{a} 
\hbox{ for an } a\in \{1,\cdots, k\} \\ 0 & \text{ otherwise } 
\end{array}\right. \cr
\bfA^{\epsilon}_{ij} &= \left\{ \begin{array}{cc} 1 & \text{ if } \{i,j\}
\in E \text{ and } i\in C_{a} \text{ and } j\in C_{b} \text{ for some } 
a\neq b \\ 0 & \text{ otherwise. } \end{array}\right.
\end{align*}
Assuming that $q<< p$, there are much fewer intercluster than intracluster edges, and thus $\bfA^{\epsilon}$ is much sparser than $\bfA^{0}$. Observe further that we can regard $\bfA^{0}$ as the adjacency matrix of a subgraph $G^{0}\subset G$, obtained by removing all the intercluster edges. In fact, $G^{0} \in \mathcal{G}(n,k,p,0)$. Let $\bfL^{0}$ denote the Laplacian of the underlying graph $G^{0}$. It is tempting to assume that $\bfL = \bfL^{0} + \bfL^{\epsilon}$, where $\bfL^{\epsilon}$ is a Laplacian associated to $A^{\epsilon}$, but unfortunately this is 
not the case. However, one can still regard $\bfL$  as a small perturbation of $\bfL^{0}$, i.e. $\bfL= \bfL^{0} + \bfE$. In order to establish bounds on the restricted isometry constants of $\bfL$, we need to bound the size of the perturbation $\bfE$, which we do in Theorem \ref{thm:EBoundsMain}. First,  recall that, as in \S \ref{sec:RandomGraphTheory}, for any $i\in C_a$, $d^{0}_{i}$ denotes the in-community degree (equivalently 
the degree of $i$ in $G_{C_a}$ or $G^{0}$) and $d^{\epsilon}_{i}$ denotes the out-of-community degree. In terms of $\bfA^{0}$ and $\bfA^{\epsilon}$, $d_{i}^0= \sum_{j}\bfA^{0}_{ij}$ and $d^{\epsilon}_{i} = \sum_{j} \bfA^{\epsilon}_{j}$. Define $r_{i} = d^{\epsilon}_i/d^{0}_i$. We can assume that $d^{\epsilon}_{i} << d^{0}_{i}$, equivalently $r_i<<1$, for all $i$. 

\begin{theorem}
\label{thm:EBoundsMain}
Let $\bfL$ be the Laplacian of a graph drawn at random from 
$\mathcal{G}(n,k,p,q)$ with $p \geq 4k(\ln(n))^2/n$ and let $\alpha$ be as in part 2 of Theorem \ref{lemma:AlphaDegBound}. Suppose further that $\max_{1\leq i \leq n}r_i  \leq r<< 1$. Then
$\bfL = \bfL^{0} 
+ \bfE$ where $\bfL^{0}$ is the Laplacian of $G^{0}$ and $\bfE = \bfE^{1} + \bfE^2$ such that $\supp(\bfE^1) = \supp(\bfA^0)$ and 
$\supp(\bfE^2) = \supp(\bfA^{\epsilon})$ and the following bounds hold:
\begin{enumerate}
\item $\|\bfE^{1}\|_{\infty}, \|\bfE^{2}\|_{\infty} \leq r + \mathcal{O}(r^2)$.
\item $\|\bfE^{1}\|_{1}, \|\bfE^{2}\|_{1} \leq \beta^{2} r + \mathcal{O}(r^2)$
\item $\|\bfE^{1}\|_{2}, \|\bfE^{2}\|_{2} \leq \beta r + \mathcal{O}(r^2)$.
\item for any $s < n$, $\|\bfE^{1}\|_{2,s}, \|\bfE^{2}\|_{2,s} \leq \beta r + \mathcal{O}(r^2) $ 
and $\|\bfE\|_{2,s} \leq 2\beta r + \mathcal{O}(r^2)$ 
\end{enumerate}
where $\beta^2 = (1+\alpha)/(1-\alpha)$.
\end{theorem}
The proof of Theorem~\ref{thm:EBoundsMain} uses the following lemma:
\begin{lemma}
\label{lemma:ComputeLaplError}
Suppose that $\bfL$ is the Laplacian of a graph $G$ drawn from $\mathcal{G}(n,k,p,q)$, and let $\bfL^{0}$ denote the 
Laplacian of the underlying (disconnected) graph $G^{0}$, thought of as drawn from $\mathcal{G}(n,k,p,0)$. Suppose further that $\max_{1\leq i \leq n}r_i \leq r<< 1$. Then $\bfL = \bfL^{0} + E$, where:
\begin{equation*}
\bfE_{ij} = \left\{\begin{array}{cc} \frac{1}{d^{0}_{i}}\left(r_i + \mathcal{O}(r^2)\right)\mathbf{A}^{0}_{ij} & \text{ if } i,j \in C_{a} \\ \frac{1}{d_{i}^{0}}\left((r_{i}-1)+ \mathcal{O}(r^2)\right)\bfA^{\epsilon}_{ij} & \text{ if } i\in C_{a} \text{ and } j\in C_{b} \text{ for } a\neq b \end{array}\right.
\end{equation*}
\end{lemma}
\begin{proof}
By definition $\bfL_{ij} = \delta_{ij} - \frac{\bfA_{ij}}{d_{i}}$ where $\delta_{ij} = 1$ if $i = j$ and zero otherwise. Observe that $\displaystyle \frac{1}{d_i} = 
\frac{1}{d^{0}_i + d^{\epsilon}_{i}} = \frac{1}{d^{0}_{i}}\frac{1}{1 + \frac{d^{\epsilon}_{i}}{d^{0}_{i}}}  = 
\frac{1}{d^{0}_{i}}\frac{1}{1 + r_i}$
Because $r_i << 1$ for all $i$, by Taylor's Theorem  $\displaystyle \frac{1}{d_i} = 
\frac{1}{d^{0}_{i}}\left( 1 - 
r_i + \mathcal{O}(r_i^2) \right) $. Direct computation reveals that, for $i,j\in C_a$:
\begin{align*}
\bfL_{ij} & = \delta_{ij} - \frac{1}{d^{0}_{i}}\left( 1 - 
r_i + \mathcal{O}(r_i^2) \right)\mathbf{A}^{0}_{ij} = \left( \delta_{ij} - \frac{\bfA^{0}_{ij}}{d^{0}_{i}}\right) + \frac{1}{d^{0}_{i}}\left(r_i + \mathcal{O}(r^2)\right)\mathbf{A}^{0}_{ij}\\
& = \mathbf{L}^{0}_{ij} + \frac{1}{d^{0}_{i}}\left(r_i + \mathcal{O}(r^2)\right)\mathbf{A}^{0}_{ij}
\end{align*}
as $1/d^{0}_{i} \leq r_{i} = \mathcal{O}(r)$. The calculation for $i,j$ in different clusters is similar.
\end{proof}

Let us write $\bfE = \bfE^{1} +\bf E^{2}$ where $\bfE^{1}$ will contain the `on-diagonal' perturbation, and $\bfE^{2}$ will 
contain the `off diagonal' perturbation:
\begin{equation*}
\bfE^{1}_{ij} = \left\{\begin{array}{cc} \mathbf{E}_{ij} & \text{ if } i,j \in 
C_{a}, \ a\in [k] \\ 0 & \text{ otherwise } \end{array}\right. \text{ and }  \bfE^{2}_{ij} = \left\{\begin{array}{cc} \mathbf{E}_{ij} & \text{ if } i\in C_{a} \text{ and } j\in C_{b} \text
{ for } a\neq b, \ a,b\in [k] \\ 0 & \text{ otherwise } \end{array}\right.
\end{equation*}

We are now ready to prove Theorem~\ref{thm:EBoundsMain}. 
 
\begin{proof}[Proof of theorem \ref{thm:EBoundsMain}] 
By definition, $\|\cdot\|_{\infty}$ is the maximum absolute row sum of the matrix in question. For any $i$, 
let $C_{a}$ be the cluster to which it belongs. Then:
\begin{equation*}
\sum_{j} |\bfE^{1}_{ij}| = \sum_{j\in C_{a}}\frac{r_i + \mathcal{O}(r^2)}{d^{0}_{i}}\mathbf{A}^{0}_{ij} = 
\frac{r_i + \mathcal{O}(r^2)}{d^{0}_{i}}\sum_{j\in C_{a}}\bfA^{0}_{ij}  = \frac{r_i + \mathcal{O}(r^2)}{d^{0}_{i}} (d_i^{0})= r_{i} + \mathcal{O}(r^{2}). 
\end{equation*}
Hence, we have 
\begin{equation*}
\|\bfE^{1}\|_{\infty} = \max_{i}\left(\sum_{j} |\bfE^{1}_{ij}|\right) = \max_{i}(r_{i}) +\mathcal{O}(r^2) \leq  
r + \mathcal{O}(r^{2}). 
\end{equation*}
Similarly:
\begin{equation*}
\sum_{j} |\bfE^{2}_{ij}| = \sum_{j\notin C_{a}}\frac{\left|(r_{i}-1)+ \mathcal{O}(r^2)\right|}{d_{i}^{0}}\bfA^{\epsilon}_{ij} \leq \frac{1+\mathcal{O}(r^2)}{d_i^{0}}\sum_{j\notin C_a}\bfA^{\epsilon}_{ij} = \frac{1+\mathcal{O}(r^2)}{d_i^{0}}(d^{\epsilon}_{i}) \leq r_i + \mathcal{O}(r^2)
\end{equation*}
as $r_i := d^{\epsilon}_{i}/d^{0}_{i}$, and so again:
\begin{equation*}
\|\bfE^{2}\|_{\infty} = \max_{i}\left(\sum_{j} |\bfE^{2}_{ij}|\right) = \max_{i}(r_{i}) +\mathcal{O}(r^2) \leq r 
+\mathcal{O}(r^2)
\end{equation*}
Before continuing, recall that as $p \geq 4k(\ln(n))^2/n$ we may choose an $\alpha = o(1)$ with respect to $n$ such that $(1-\alpha)n_0p \leq d^{0}_{i} \leq (1+\alpha)n_0p$ almost surely for all $i\in [n]$, by part 2 of Theorem \ref{lemma:AlphaDegBound}. Now, because $\|\cdot\|_{1}$ is the maximum absolute column sum, for any $j$ in cluster $C_{a}$, 
we have
\begin{align*}
\sum_{1\leq i\leq n}|\bfE^{1}_{ij}| &= \sum_{i\in C_{a}}\frac{r_i + \mathcal{O}(r^2)}{d^{0}_{i}}\mathbf{A}^{0}_{ij}
 \leq \frac{r+\mathcal{O}(r^2)}{d^{0}_{\min}}\sum_{i\in C_{a}}\bfA^{0}_{ij}  \leq \frac{r+ \mathcal{O}(r^2)}{d^{0}_{\min}}(d^{0}_{\max}) \\
 & \leq (r + \mathcal{O}(r^2))\frac{(1+\alpha)n_0p}{(1-\alpha)n_0p} = \frac{r(1+\alpha)}{(1-\alpha)} + \mathcal{O}(r^2) = \beta^2 r + \mathcal{O}(r^2)
\end{align*}
and so:
\begin{equation*}
\|\bfE^{1}\|_{1} = \max_{j}\left(\sum_{i} |\bfE^{1}_{ij}|\right) \leq  \beta^2 r + \mathcal{O}(r^2). 
\end{equation*}
Similarly;
\begin{align*}
\sum_{1\leq i \leq n} |\bfE^{2}_{ij}| &= \sum_{i\notin C_{a}}\frac{\left|(r_{i}-1)+ \mathcal{O}(r^2)\right|}{d_{i}^{0}}\bfA^{\epsilon}_{ij} \leq \frac{1 + \mathcal{O}(r^2)}{d^{0}_{\min}}\sum_{i\notin C_{a}}\bfA^{\epsilon}_{ij}  = (1+\mathcal{O}(r^2))\frac{d^{\epsilon}_{j}}{d^{0}_{\text{min}}} \\
& \leq (1+\mathcal{O}(r^2))\frac{1+\alpha}{1-\alpha}\frac{d^{\epsilon}_{j}}{d^{0}_{\text{max}}} \leq (1+\mathcal{O}(r^2))\frac{1+\alpha}{1-\alpha}\frac{d^{\epsilon}_{j}}{d^{0}_{j}}  \leq  \beta^2 r_j + \mathcal{O}(r^2)\\
\end{align*}
and so again we have
\begin{equation*}
\|\bfE^{2}\|_{1} = \max_{j}\left(\sum_{i} |\bfE^{2}_{ij}|\right) \leq \beta^2 \max_{j}(r_j) + 
\mathcal{O}(r^2) \leq \beta^2 r  + \mathcal{O}(r^2). 
\end{equation*}
Now, by the Riesz-Thorin interpolation theorem:
\begin{equation*}
\|\bfE^{1}\|_{2} \leq \|\bfE^{1}\|_{1}^{1/2}\|\bfE^{1}\|_{\infty}^{1/2} \leq \left(\beta^2 r + \mathcal{O}(r^2)\right)^{1/2}(r + \mathcal{O}(r^2))^{1/2} = \beta r +\mathcal{O}(r^2)
\end{equation*} 
and by an identical argument $\|\bfE^{2}\|_{2} \leq \beta r + \mathcal{O}(r^2)$. Part 4 follows from the simple fact that for any matrix $\mathbf{B}$, $\|\mathbf{B}\|_{2,s} \leq \|\mathbf{B}\|_{2}$, and the bound for $\|\bfE\|_{2,s}$ follows from the fact that $\bfE = \bfE^{1} + \bfE^{2}$ and the triangle inequality.
\end{proof}

We are now ready to state and prove the main result of this section. 
\begin{theorem}
\label{thm:Twostory}
Suppose that $G \in \mathcal{G}(n,k,p,q)$, with block sizes $n_0 = n/k$ and Laplacian $\bfL$. 
Suppose further that:
\begin{enumerate}
\item $p >>  (ln(n_0))^2/n_0$ 
\item $k$ is $O(1)$ with respect to $n$.
\item $\displaystyle\max_{1\leq i\leq n} r_i \leq r << 1$.
\end{enumerate}
 If $t = \gamma n_0$ with $\gamma \in (0,1)$ then:
\begin{equation}
\delta_{t}(\bfL) \leq \gamma + D_1r + \frac{D_2}{\sqrt{n_0}} + D_3 \frac{r}{\sqrt{n_0}} + \mathcal{O}(r^2) + 
o(\frac{1}{\sqrt{n_0}}) 
\label{eq:Twostory}
\end{equation} 
almost surely, where $D_1$, $D_2$ and $D_3$ are $\mathcal{O}(1)$ with respect to $n_0$
\end{theorem}
\begin{proof}
As above, we write $\bfL = \bfL^{0} + \bfE$, where $\bfL^{0}$ is the Laplacian of the 
subgraph $G^{0}$, thought of as drawn from $\mathcal{G}(n,k,p,0)$. We define $\delta_{t} := \delta_{t}(\bfL^{0})$ 
and $\hat{\delta}_{t} := \delta_{t}(\bfL)$. By Theorem~\ref{thm:PerturbedRIC}:
\begin{equation}
\label{eq:deltaTilde}
\hat{\delta}_{t} \leq (1+\delta_{t})(1 + \epsilon^{t}_{\bfL})^{2} - 1.
\end{equation}
As defined in a previous section, $\epsilon^{t}_{\bfL} := \|\bfE\|_{2,t} / \|\bfL^{0}\|_{2,t}$.  
Theorem~\ref{thm:EBoundsMain} gives us that $\|\bfE\|_{2,t} \leq 2\beta r + \mathcal{O}(r^2)$, while Lemma~\ref{lemma:LoweBoundSemiNorm2} gives $\|\bfL^{0}\|_{2,t} 
\geq \max_{1\leq a\leq k}\lambda_{t-1}(\bfL^{a})$ (recall that $\bfL^{a}$ denotes the Laplacian of the subgraph $G_{C_a}$). By a similar argument to that of Theorem \ref{thm:RIPBoundq0},    
we can apply Lemma \ref{theorem:ChungLuVu} (and see also remark \ref{remark:LambdaMin}) to get:
\begin{equation*}
\max_{1\leq a\leq k}\lambda_{t-1}(\bfL^{a}) \geq 1 - \frac{4p^{-1/2}}{\sqrt{n_0}} - o(\frac{1}{\sqrt{n_0}})
\end{equation*}
almost surely, assuming that $k$ is eventually bounded with respect to $n_0$. For convenience, choose $n_0$ large enough 
such that $4p^{-1/2}/\sqrt{n_0} + o(\frac{1}{\sqrt{n_0}}) \leq 1/4$. Under this assumption:
\begin{equation*}
\epsilon^{t}_{\bfL^{0}}  =  \frac{\|\bfE\|_{2,t} }{ \|\bfL^{0}\|_{2,t}} \leq 
\frac{2\beta r + \mathcal{O}(r^2)}{1 - 1/4} = \frac{8\beta r}{3} + \mathcal{O}(r^2)
\end{equation*}
Combining this with \eqref{eq:deltaTilde} and Theorem \ref{thm:RIPBoundq0}:
\begin{align*}
\hat{\delta}_{s} &\leq \left( 1 + \gamma +(1-\gamma)\frac{8p^{-1/2}}{\sqrt{n_0}} + 
o(\frac{1}{\sqrt{n_0}})\right)\left(1 + \frac{8\beta r}{3} + \mathcal{O}(r^2)\right)^2 - 1\\
& = \left( 1 + \gamma +(1-\gamma)\frac{8p^{-1/2}}{\sqrt{n_0}} + o(\frac{1}{\sqrt{n_0}})\right) 
\left(1 + \frac{16\beta r}{3} + \mathcal{O}(r^2)\right) - 1\\
& = \gamma + r\left(\frac{16\beta(1+\gamma)}{3}\right) + \frac{1}{\sqrt{n_0}}\left(8(1-\gamma)p^{-1/2}\right) 
+ \frac{r}{\sqrt{n_0}}\left(\frac{128(1-\gamma)p^{-1/2}\beta}{3} \right)  +\mathcal{O}(r^2) + o(\frac{1}{\sqrt{n_0}})
\end{align*}
and letting $D_1 = 16\beta(1+\gamma)/3 $, $D_2 = 8(1-\gamma)p^{-1/2}$, $D_3 = 128(1-\gamma)p^{-1/2}\beta^{2}/3$ 
gives equation \eqref{eq:Twostory}. Finally, because $\beta \to 1$ as $n \to \infty$, we indeed have that $D_1$, 
$D_2$ and $D_3$ are $\mathcal{O}(1)$ with respect to $n_0$.
\end{proof}

\section{The Coherence Properties for Laplacians of Random Graphs}
\label{section:CoherenceProperties}
In this section we study the coherence properties of Laplacians of random graphs. Mainly, we compute a lower bound for $|\langle \ell_{i},\ell_{j}\rangle |$ when $i$ and $j$ are in the same cluster and hence a lower bound for $\mu(\bfL_{-1})$. We begin with a basic result for graphs from $\mathcal{G}(n,k,p,0)$.
\begin{theorem}
Let $\bfA$ be the adjacency matrix of a graph $G^0$ drawn at random from $\mathcal{G}(n,k,p,0)$ and 
define $\chi_{i,j} = \sum_{k} \bfA_{ik}\bfA_{kj}$. Suppose further that there exists an $\alpha \in (0,1)$ such that $(1-\alpha)pn_0 \leq d_{i}^{0} \leq (1+\alpha)pn_0$. Then for any $i,j \in C_{1}$ with $i \neq j$:
\begin{equation}
(1-\alpha)n_0p^2 - \sqrt{\frac{(1+\alpha)n_0p\ln(1/\delta)}{2}} \leq 
\chi_{i,j} \leq (1+\alpha)n_0p^2 + \sqrt{\frac{(1+\alpha)n_0p\ln(1/\delta)}{2}}
\end{equation}
with probability $1 - \delta$.
\label{thm:chi1bound}
\end{theorem}
\begin{proof}
Let $B_{1}(i)$ denote the ball of radius $1$ centered at $i$. 
That is, $B_{1}(i) := \{k \in [n]:\ \{i,k\}\in E\}$. 
By definition $|B_1(i)| = d_i$. Observe that:
\begin{equation*}
\chi_{i,j} = \sum_{k} \bfA_{ik}\bfA_{kj} = \sum_{k\in B_1(i)}\bfA_{kj}
\end{equation*}
Each $\bfA_{kj}$ is an i.i.d. Bernoulli random variable with parameter $p$. Applying the Chernoff bound to this sum:
\begin{equation*}
\mathbb{P}\left[\sum_{k\in B_1(i)}\bfA_{kj} \leq pd_i - \sqrt{\frac{d_i\ln(1/\delta)}{2}}\right] \leq 
\exp\left(-\frac{2
\left(\sqrt{\frac{d_i\ln(1/\delta)}{2}}\right)^{2}}{d_i}\right) = \delta \\
\end{equation*}
and $(1-\alpha)pn_0 \leq d_i \leq (1+\alpha)pn_0$ so:
\begin{align*}
\mathbb{P}\left[ \chi_{i,j}\leq (1-\alpha)n_0p^2 - \sqrt{\frac{(1+\alpha)n_0p\ln(1/\delta)}{2}} \right] 
\leq \mathbb{P}\left[ \chi_{i,j}\leq pd_i - \sqrt{\frac{d_i\ln(1/\delta)}{2}} \right] \leq \delta
\end{align*}
The upper bound on $\chi_{i,j}$ is proved analogously.
\end{proof}

We now provide a bound on the inner product $|\langle \ell^0_i,\ell^0_j\rangle |$ when $i$ and $j$ are in the same cluster:
\begin{theorem}
\label{thm:singleL1bound}
Let $\bfL^{0}$ be the Laplacian of a graph $G^{0}$ drawn at random from 
$\mathcal{G}(n,k,p,0)$. As in Theorem \ref{thm:chi1bound}, suppose that 
there exists an $\alpha \in (0,1)$ such that $(1-\alpha)pn_0 \leq 
d_{i}^{0} \leq (1+\alpha)pn_0$ but in addition assume that $\alpha$ is $o(1)$ with respect to $n$. Define $\beta^2 = (1+\alpha)/(1-\alpha)$. Then for any $\delta\in (0, 1)$
 and for any $i,j\in C_a$, $a\in [k]$, with 
$i\neq j$:
\begin{equation*}
|\langle \ell^{0}_{i},\ell^{0}_{j}\rangle | \geq \frac{1}{n_0}\left(\frac{\beta^2}{(1+\alpha)}\right)- \frac{1}{n_0^{3/2}}\sqrt{\frac{\ln(1/\delta)}{2(1+\alpha)^3p^3}}
\end{equation*}
with probability $1-\delta$ for $n$ large enough.
\end{theorem}
\begin{proof}
Recall that
\begin{equation*}
\left|\langle \ell^{0}_{i},\ell^{0}_{j}\rangle \right| = \left|- \left(\frac{1}{d_i} + \frac{1}{d_j}\right) \bfA_{ij} + \sum_{k=1}^n\frac{\bfA_{ik}\bfA_{jk}}{d_{k}^2} \right|
\end{equation*}
There are two possibilities; either $i$ and $j$ are connected ($\{i,j\}\in E$) or they are not ($\{i,j\}\notin E$). We treat these cases separately. Suppose first that $\{i,j\}\notin E$. In this case, $\bfA_{ij} = 0$ and so:
\begin{equation*}
\left|\langle \ell^{0}_{i},\ell^{0}_{j}\rangle \right| = \sum_{k=1}^n\frac{\bfA_{ik}\bfA_{jk}}{d_{k}^2} \geq \frac{1}{d^{2}_{\max}} \chi_{ij}
\end{equation*}
by assumption, $d_{\max} \leq (1+\alpha)pn_0$ and by so by Theorem \ref{thm:chi1bound}:
\begin{align}
\left|\langle \ell^{0}_{i},\ell^{0}_{j}\rangle \right| & \geq \frac{1}{(1+\alpha)^2p^2n_0^2} \left((1-\alpha)n_0p^2 - \sqrt{-\frac{(1+\alpha)n_0p\ln(\delta)}{2}}\right)\nonumber \\
	&= \frac{1}{n_0}\left(\frac{1}{(1+\alpha)\beta^2}\right) - \frac{1}{n_0^{3/2}}\sqrt{\frac{-\ln(\delta)}{2(1+\alpha)^3p^3}} \label{eq:LeadingTerm1}
\end{align}
Alternatively, suppose that $\{i,j\}\in E$. Then $\bfA_{ij} = 1$ and by Theorem~\ref{thm:chi1bound}:
\begin{align}
|\langle \ell^{0}_{i},\ell^{0}_{j}\rangle | &\geq \frac{2}{d_{\max}} - \frac{1}{d^{2}_{\min}}\chi_{ij} \geq \frac{2}{(1+\alpha)pn_0} - \frac{1}{(1-\alpha)^2p^2n_0^2}\chi_{ij} \nonumber \\
	& \geq \frac{2}{(1+\alpha)pn_0} - \frac{1}{(1-\alpha)^2p^2n_0^2}\left((1+\alpha)n_0p^2 + \sqrt{-\frac{(1+\alpha)n_0p\ln(\delta)}{2}}\right) \quad  \nonumber \\
    & = \frac{1}{n_0}\left(\frac{2}{(1-\alpha)p} - \frac{\beta^2}{(1+\alpha)}\right) - \frac{1}{n_0^{3/2}}\sqrt{\frac{-\ln(\delta)\beta^2}{2(1-\alpha)^3p^3}} \label{eq:LeadingTerm2}
\end{align}
Let us compare the leading terms in \eqref{eq:LeadingTerm1} and \eqref{eq:LeadingTerm2}. Because $\alpha \to 0$ as $n\to \infty$, $\beta \to 1$ and thus $1/((1+\alpha)\beta) \to 1$. On the other hand, $2/((1-\alpha)p) - \beta^2/(1-\alpha) \to 2/p - 1 > 1$, assuming $p <1$. Thus, for $n$ large enough, the right hand side of \eqref{eq:LeadingTerm2} is larger than the right hand side of \eqref{eq:LeadingTerm1}, and so we can summarize both bounds as:
\begin{equation*}
|\langle \ell^{0}_{i},\ell^{0}_{j}\rangle | \geq \frac{1}{n_0}\left(\frac{\beta^2}{(1+\alpha)}\right) - \frac{1}{n_0^{3/2}}\sqrt{\frac{-\ln(\delta)}{2(1+\alpha)^3p^3}}
\end{equation*}
for large enough $n$.
\end{proof}

Next we consider the case where intercluster edges are present, i.e. where $q > 0$.
\begin{theorem}
\label{thm:NoisyInnerProdBound}
Let $\bfL$ be the Laplacian of $G$ drawn from $\mathcal{G}(n,k,p,q), q>0$, and suppose that $i,j\in C_{1}$. 
Assume further that:
\begin{enumerate}
\item $p \geq 4k(\ln(n))^2/n$
\item $\max_{i} r_{i} \leq r$ with $r =r_0/\sqrt{n_0}$ where $r_0 = \mathcal{O}(1)$ with respect to $n_0$.
\end{enumerate}
Then almost surely:
\begin{equation*}
|\langle \ell_{i},\ell_{j}\rangle | \geq \frac{1}{n_0}\left(\frac{\beta^2}{1+\alpha}\right) - 
o(\frac{1}{n_0}). 
\end{equation*}
where $\alpha$ is $o(1)$ with respect to $n$ and $\beta^2 = (1+\alpha)/(1-\alpha)$.
\end{theorem}

\begin{proof}
As before, let $\bfe_{i}^{1}$ (resp. $\bfe_{j}^{1}$) denote the $i$-th (resp. $j$-th) column of $\bfE^{1}$, while $\bfe_{i}^{2}$ (resp. $\bfe_{j}^{2}$) denote the $i$-th (resp. $j$-th) column of $\bfE^{2}$. Then:
\begin{align*}
\langle \ell_{i},\ell_{j}\rangle &= \langle \ell^{0}_{i} + \bfe^{1}_{i} + \bfe^{2}_{i},\ell^{0}_{j} + \bfe^{1}_{j} + \bfe^{2}_{j}\rangle \\
 & = \langle \ell^{0}_{i}, \ell^{0}_{j} \rangle + \langle \ell^{0}_{i}, \bfe_{j}^{1} \rangle + \langle \ell^{0}_{i}, \bfe_{j}^{2} \rangle  + \langle \bfe^{1}_{i}, \ell^{0}_{j} \rangle + \langle \bfe^{1}_{i}, \bfe^{1}_{j} \rangle + \langle \bfe^{1}_{i}, \bfe^{2}_{j}\rangle + \langle \bfe^{2}_{i}, \ell^{0}_{j} \rangle + \langle \bfe^{2}_{i}, \bfe^{1}_{j} \rangle + \langle \bfe^{2}_{i}, \bfe^{2}_{j} \rangle
\end{align*}
By construction, $\bfe^{1}_{i}$ and $\bfe^{2}_{j}$ have disjoint support (and similarly for $\bfe^{1}_{j}$ and 
$\bfe^{2}_{i}$), as do $\ell^{0}_{i}$ and $\bfe^{2}_{j}$ (and similarly $\ell^{0}_{j}$ and $\bfe^{2}_{i}$). 
Hence:
\begin{equation*}
\langle \ell^{0}_{i}, \bfe_{j}^{2} \rangle = \langle \bfe^{1}_{i}, \bfe^{2}_{j}\rangle = \langle \bfe^{2}_{i}, \ell^{0}_{j} \rangle = \langle \bfe^{2}_{i}, \bfe^{1}_{j} \rangle  = 0
\end{equation*}
and so:
\begin{equation}
\label{eq:InnerProdAbsBound}
\left|\langle \ell_{i},\ell_{j}\rangle \right| \geq \left|\langle \ell^{0}_{i},\ell^{0}_{j}\rangle \right| - \left|\langle \ell_{i}^{0},\bfe^{1}_{j} \rangle \right| - \left|\langle \bfe^{1}_{i},\ell^{0}_{j} \rangle \right| - \left|\langle \bfe^{1}_{i},\bfe^{2}_{j}\rangle \right| - \left|\langle \bfe_{i}^{2},\bfe_{j}^{2}\rangle \right|
\end{equation}
Now $\ell_{i}^{0}$ and $\ell_{j}^{0}$ can be thought of as columns in the Laplacian of a graph drawn at random from $\mathcal{G}(n,k,p,0)$. By part 2 of Theorem \ref{lemma:AlphaDegBound} we may choose an $\alpha = o(1)$ with respect to $n$ such that $(1-\alpha)pn_0 \leq d^{0}_{i} \leq (1+\alpha)pn_0$ for all $i$, with probability $1-1/n$. Now we apply Theorem \ref{thm:singleL1bound} to get:
\begin{equation*}
|\langle \ell_{i}^{0},\ell_{j}^{0}\rangle | \geq \frac{1}{n_0}\left(\frac{\beta^2}{(1+\alpha)}\right)- \frac{1}{n_0^{3/2}}\sqrt{\frac{-\ln(\delta)}{2(1+\alpha)^3p^3}}
\end{equation*}
with probability $1-\delta - 1/n$. Taking $\delta = e^{-p^{3}\sqrt{n_0}}$ (and noting $e^{-p^{3}\sqrt{n_0}} < 1/n_0$ for large enough $n_0$) 
\begin{align*}
\frac{1}{n_0}\left(\frac{\beta^2}{(1+\alpha)}\right)- \frac{1}{n_0^{3/2}}\sqrt{\frac{-\ln(\delta)}{2(1+\alpha)^3p^3}} &= \frac{1}{n_0}\left(\frac{\beta^2}{(1+\alpha)}\right)- \frac{1}{n_0^{3/2}}\sqrt{\frac{n_0^{1/2}}{2(1+\alpha)^3}} \\
& = \frac{1}{n_0}\left(\frac{\beta^2}{1+\alpha}\right) - \frac{1}{n_0^{5/4}}\frac{1}{\sqrt{2}(1+\alpha)^{3/2}} = \frac{1}{n_0}\left(\frac{\beta^2}{1+\alpha}\right) - o(1/n_0)
\end{align*}
with probability at least $1 - 1/n_0 - 1/n = 1 - (k+1)/n$. Next, we consider the term $\langle \ell_{i}^{0},\bfe_{j}^{1}\rangle$.
\begin{align*}
|\langle \ell_{i}^{0},\bfe_{j}^{1}\rangle| &= \left|\sum_{m}\left(\delta_{im} - \frac{\bfA^{0}_{im}}{d^{0}_m}\right)\left(\frac{r_m\bfA^{0}_{mj}}{d^{0}_m} +\mathcal{O}(r^3)\right)\right| = \left|\frac{r_i\bfA^{0}_{ij}}{d^{0}_{i}} - \sum_{m}\frac{r_m \bfA^{0}_{im}\bfA^{0}_{mj}}{(d^{0}_{m})^{2}} + \mathcal{O}(r^3)\right| \\
 & \leq \frac{r}{d^{0}_{\min}} + \frac{r}{(d^{0}_{\min})^{2}}\sum_{m}\bfA^{0}_{im}\bfA^{0}_{mj} +\mathcal{O}(r^3)  \leq \frac{r}{d^{0}_{\min}} + \frac{r}{(d^{0}_{\min})^{2}}\min\{d^{0}_{i},d^{0}_{j}\} +\mathcal{O}(r^3) \\
 & \leq \frac{r}{d^{0}_{\min}}\left(1 + \frac{d^{0}_{\max}}{d^{0}_{\min}}\right) + \mathcal{O}(r^3)
\end{align*}
As stated earlier, $d^{0}_{\min} \geq (1-\alpha)pn_0$ and $d^{0}_{\max} \leq (1+\alpha)pn_0$ with probability $1 - 1/n$, and so:
\begin{equation*}
|\langle \ell_{i}^{0},\bfe_{j}^{1}\rangle| \leq \frac{r}{(1-\alpha)pn_0}\left(1 + \frac{(1+\alpha)pn_0}{(1-\alpha)pn_0}\right) + \mathcal{O}(r^3) = \frac{r}{(1-\alpha)pn_0}\left(1 + \beta^2 \right)  + \mathcal{O}(r^3)
\end{equation*}
Clearly the same bound holds for $|\langle \bfe^{1}_{i},\ell^{0}_{j}\rangle |$. Finally, consider $|\langle \bfe_{i}^{1},\bfe_j^{1}\rangle |$. A similar calculation to the one above reveals that:
\begin{equation*}
|\langle \bfe_{i}^{1},\bfe_j^{1}\rangle | \leq \frac{r^2}{n_0}\left(\frac{\beta^4}{(1-\alpha)p}\right) + \mathcal{O}(r^4) 
\end{equation*}
as well as:
\begin{equation*}
|\langle \bfe_{i}^{2},\bfe_{j}^2\rangle | \leq \frac{r}{(1-\alpha)pn_0} + \mathcal{O}(r^3)
\end{equation*}
Putting all these bounds back in to equation \eqref{eq:InnerProdAbsBound}, we get:
\begin{align*}
|\langle \ell_i,\ell_j \rangle | & \geq \frac{1}{n_0}\left(\frac{\beta^2}{1+\alpha}\right) - o(\frac{1}{n_0}) - \left( \frac{2r(1 + \beta^2)}{(1-\alpha)pn_0}  + \mathcal{O}(r^3) + \frac{r^2}{n_0}\left(\frac{\beta^4}{(1-\alpha)p}\right) + \mathcal{O}(r^4) + \frac{r}{(1-\alpha)pn_0} + \mathcal{O}(r^3)\right) \\
	& = \frac{1}{n_0}\left(\frac{\beta^2}{1+\alpha}\right) - \frac{r}{n_0}\left(\frac{2(1 + \beta^2)}{(1-\alpha)p} + \frac{r\beta^4}{(1-\alpha)p} + \frac{1}{(1-\alpha)p}\right) + \mathcal{O}(r^3)
\end{align*}
Using the assumption that $r = r_0/\sqrt{n_0}$, we see that $r/n_0 = r_0/n_0^{3/2}$ which is of order $o(1/n_0)$. Similarly $r^3 = r_0^3/n_0^{3/2}$ which is also of order $o(1/n_0)$. Thus:
\begin{equation*}
|\langle \ell_i,\ell_j \rangle | \geq \frac{1}{n_0}\left(\frac{\beta^2}{1+\alpha}\right) - o(\frac{1}{n_0})
\end{equation*}
with probability $1- \mathcal{O}(1/n_0)$, as claimed.
\end{proof}

\begin{remark}
\label{remark:coherence}
Since the coherence  $\mu$ of $\bfL$ is defined as $\mu = \max_{i\neq j}|\langle \ell_i,\ell_j\rangle |$. 
Theorem \ref{thm:NoisyInnerProdBound} also implies that$\mu \geq \frac{1}{n_0}\left(\frac{\beta^2}{1+\alpha}\right) - o(\frac{1}{n_0})$.
\end{remark}

\section{Compressive Sensing Clustering for Graphs from $\mathcal{G}(n,k,p,0)$}
\label{section:CompClustforq_zero}
In this section, we consider graphs from $\mathcal{G}(n,k,p,0)$. 
This is the case that there are no inter-cluster edges. 
We may safely assume that each $G_{C_i}$ is 
connected, as long as $p > \ln(n_0)/n_0$.

\begin{lemma}
Let $\bfL$ be the Laplacian of a graph with $k$ connected components 
$C_1,\ldots, C_k$. Assume, without loss of generality, that vertex 
$1\in C_{1}$. Define $\bfx^{\#}$ as: 
\begin{equation}
\label{eq:ConnComponent1}
\bfx^{\#} := \argmin \{ \|\bfx\|_{0}: \ \bfx\in\mathbb{R}^{n}, \  
\bfL\bfx = \mathbf{0} \text{ and } x_{1} = 1\}
\end{equation}
Then $\bfx^{\#} = \bfone_{C_{1}}$
\end{lemma}
\begin{proof}
From Theorem  \ref{thm:IndicatorVecsKernel}, 
$\ker(\bfL) = \text{span}\{\bfone_{C_1},\ldots, \bfone_{C_k}\}$ and 
hence any $\bfx$ with $\bfL\bfx = \mathbf{0}$ can be written as: 
$\bfx = \sum_{i=1}^{k}c_{i}\bfone_{C_i}$. If $x_1 = 1$ 
(recall we are assuming that $1\in C_1$) then $c_{1} = 1$ and so $\bfx 
= \bfone_{C_{1}} + \sum_{i=2}^{k}c_{i}\bfone_{C_{i}}$. Recall that all 
the $\bfone_{C_i}$ have disjoint support, and so clearly the sparsest 
solution $\bfx$ such that $\bfL\bfx=\mathbf{0}$ is that 
$\bfx$ has $c_2 = c_3 = \ldots = c_k = 0$. Hence indeed $\bfx^{\#} = 
\bfone_{C_1}$ 
\end{proof}

One can rephrase this slightly as a standard compressed sensing 
problem. Recall $\ell_{i}$ is the $i$-th column of $\bfL$ and $\bfL_{-i} := 
[\ell_{1},\ldots,\ell_{i-1},\ell_{i+1},\ldots, \ell_{n}]$ in which case:
\begin{equation*}
x_i = 1 \quad \Leftrightarrow \quad \bfL_{-i}\bfx_{-i} = -\ell_{i}, 
\end{equation*}
where if $\bfx = [x_1,x_2,\ldots, x_n]^\top \in\mathbb{R}^{n}$, 
$\bfx_{-i} = [x_1, \ldots,x_{i-1},x_{i+1},\ldots, 
x_{n}]^{\top}\in\mathbb{R}^{n-1}$. 
Thus problem \eqref{eq:ConnComponent1} becomes:
\begin{equation}
\label{eq:ConnComponent2}
\bfx^{\#} := \argmin \{ \|\bfx\|_{0}: \ \bfx\in\mathbb{R}^{n-1}, 
\  \bfL\bfx = -\ell_{1} \}. 
\end{equation}

We now show that problem \eqref{eq:ConnComponent2} can be efficiently 
solved using OMP (Algorithm \ref{algorithm:OMP} in \S \ref{section:SparseRecovery})

\begin{theorem}
\label{thm:Successq0}
Suppose that $G \in \mathcal{G}(n,k,p,0)$ has $k$ 
connected components $C_1,\ldots, C_k \subset V$ each of size $n_0$. OMP (Algorithm~\ref{algorithm:OMP}) applied to the sparse recovery problem 
\eqref{eq:ConnComponent2} will return 
$\bfx^{\#} = \bfone_{C_1\setminus\{1\}}$ after $n_0 - 1$ iterations. 
$C_1$ can then be recovered as $\{1\}\cup \supp(\bfx^{\#})$.
\end{theorem}

\begin{remark}
Note that the assumption that $G \in \mathcal{G}(n,k,p,0)$ is not necessary; this theorem holds for all graphs with $k$ connected components. Also, observe that this theorem is not probabilistic. It holds with certainty for all graphs $G$.
\end{remark}

We shall appeal to the exact recovery condition, 
Theorem~\ref{thm:ExactRecovCond} to 
establish Theorem~\ref{thm:Successq0}. Let us begin with 

\begin{lemma}
\label{lemma:LaplacianInjective}
If $C_1\subset V$ is a connected component then 
$\bfL_{C_{1}\setminus\{1\}}$ is injective, where 
 $\bfL_{C_1\setminus\{1\}}$ is 
the submatrix from $\bfL$ with column indices in $C_1\setminus\{1\}$. 
\end{lemma}
\begin{proof}
First, observe that $\bfL_{C_1\setminus\{1\}} = \left[\begin{array}{c}  
\bfL^{1}_{-1} \\ \mathbf{0} \end{array}\right]$, 
where $\bfL^{1}$ is the Laplacian of $G_{C_1}$. 
It suffices to show that 
$\bfL^{1}_{-1}$ is injective. As $G_{C_1}$ is 
connected, $\ker(\bfL^1) = \text{span}\{\bfone\}$, 
by Theorem~\ref{thm:SpectralProp}. Suppose there exists a $\bfu \in 
\mathbb{R}^{n_0-1}, \ \bfu\neq 0$ such that $\bfL^{1}_{-1}\bfu = 
\mathbf{0}$. Then $[0, \bfu^{\top}]^{\top}$ is 
in $\text{ker}(\bfL^{1}) = 
\text{span}\{\bfone\}$, a contradiction.
\end{proof}

\begin{theorem}
\label{thm:MainConnectedComponent}
With notation as above (i.e. $C_{1}$ is the connected component of $G$ 
containing the first vertex) we have that:
\begin{equation}
\|\bfL_{C_{1}\setminus\{1\}}^{\dagger}\bfL_{C_{1}^{c}}\|_{1} = 0
\label{eq:Zero}
\end{equation}
\end{theorem}
\begin{proof}
By Lemma~\ref{lemma:LaplacianInjective} $\bfL_{C_1\setminus\{1\}}$ is 
injective so its pseudo-inverse is given by:
\begin{equation*}
\bfL_{C_1\setminus\{1\}}^{\dagger} = \left(\bfL_{C_1\setminus\{1\}}^{
\top}\bfL_{C_1\setminus\{1\}}\right)^{-1}
\bfL_{C_1\setminus\{1\}}^{\top}
\end{equation*}
As observed in the proof of Lemma~\ref{lemma:LaplacianInjective}, 
$\bfL_{C_1\setminus\{1\}} = \left[\begin{matrix} \bfL^{1}_{-1} \\ 
\mathbf{0}\end{matrix}\right]$. Similarly, 
$\bfL_{C_1^{c}} = \left[\begin{matrix} \mathbf{0} \\ 
\tilde{\bfL}\end{matrix}\right]$, 
where $\tilde{\bfL}$ denotes the Laplacian of $G_{C_1^{c}}$.  In both cases $\mathbf{0}$ denotes the zero matrix of appropriate size.
To show \eqref{eq:Zero} it will suffice to show that 
$\bfL_{C_1\setminus\{1\}}^\top \bfL_{C_1^{c}} = 0$, but this follows 
easily as:
\begin{equation*}
\left[\begin{array}{cc} (\bfL^{1}_{-1})^{\top} 
& \mathbf{0}^{\top} 
\end{array}\right]\left[\begin{matrix} \mathbf{0} \\ \tilde
{\bfL}\end{matrix}\right] = (\bfL^{1}_{-1})^{\top}
\mathbf{0} + \mathbf{0}^{\top}\tilde{\bfL} = \mathbf{0}
\end{equation*}
This completes the proof. 
\end{proof}
Theorem~\ref{thm:Successq0} now follows easily:
 
\begin{proof}[Proof of Theorem~\ref{thm:Successq0}] 
By Theorem \ref{thm:MainConnectedComponent} the exact recover condition 
(Theorem \ref{thm:ExactRecovCond}) holds , and thus the result follows. 
\end{proof}

\section{Compressive Sensing Clustering for Graphs from $\mathcal{G}(n,k,p,q)$}
\label{section:CompClustforq_nonzero}
We now turn our attention to the same problem for $\mathcal{G}(n,k,p,q)$ 
with $0<q << p$.  We follow the notation established in \S \ref{sectionRIP_q_nonzero}, and decompose $\bfA$ as $\bfA= \bfA^{0} + \bfA^{\epsilon}$ where $\bfA^{\epsilon}$ contains the intercluster edges and $\bfA^{0}$ is the adjacency matrix of the subgraph $G^{0}$ with $k$ connected components. Denoting the $i$-th column of $\bfL^{0}$ as $\ell^{0}_{i}$ 
and the $i$-th column of $\bfE$ as $\bfe_{i}$, observe that 
\begin{equation*}
\ell_{1} = \ell^{0}_{1} + \bfe_{1} = 
-\bfL^{0}\bfone_{C_1\setminus\{1\}} + 
\bfe_1 \quad \hbox{ or }  \quad 
-\ell_1 = \bfL \bfone_{C_1\setminus\{1\}} - \bfe_1. 
\end{equation*}
Defining $\bfx^{\#}$ as:
\begin{equation}
\bfx^{\#} := \argmin\{ \| 
\bfL_{-1}\bfx - (-\ell_{1})\|_{2} : \ \bfx \in \mathbb{R}^{n-1}, \  \|\bfx\|_{0}\leq n_0 \}. 
\label{eq:CompressedClustering}
\end{equation}
Then $\bfx^{\#} \approx \bfone_{C_1\setminus\{1\}}$. 
We recognize this as a totally perturbed sparse recovery problem 
\eqref{eq:SparseRecoveryTotallyPerturbed}, with $\Phi = \bfL^{0}$, 
$\hat{\Phi} = \bfL^{0} + \bfE = \bfL$, $\bfy = -\ell_1 = -\ell_1^{0} - \bfe_1$. 
We shall show that, provided $\bfE$ is small enough, 
$\text{supp}(\bfx^{\#}) = 
\text{supp}(\bfone_{C_1\setminus\{1\}}) = C_1\setminus\{1\}$. Hence, 
solving problem \eqref{eq:CompressedClustering} is equivalent to 
finding the community $C_1$. 

Unfortunately, a straightforward  application of 
Algorithm~\ref{algorithm:OMP} to \eqref{eq:CompressedClustering} does 
not work. Empirically, it is observed that the first several  greedy 
steps in Algorithm~\ref{algorithm:OMP} invariably pick up wrong indices 
due to the presence of noise $\bfE$. We need a new approach. Let us
start with the following

\begin{theorem}
\label{thm:Mckenzie}
Let $\bfL$ be the Laplacian of a graph $G\in \mathcal{G}(n,k,p,q)$ with 
$p \geq 4k(\ln(n))^2/n$. Suppose that 
\begin{equation}
\label{as}
\max_{i\in\{1,\ldots, n\}}r_i := \max_{i\in \{1, \cdots, n\}}
d^{\epsilon}_{i}/d^{0}_{i} \leq r, 
\end{equation} 
where $r = r_0/\sqrt{n_0}$ with $r_0$ a constant independent of $n_0$. 
Define $\Omega := \mathcal{L}_{\lceil 
10(n_0-1)/9\rceil}(L_{-1}^\top \ell_{1})$. 
For $n_0$ large enough, we have that $C_1\setminus \{1\}\subset 
\Omega$ 
almost surely.
\end{theorem}
\begin{proof}
Suppose otherwise, then there exists an $i^{*}\in C_{1}\setminus\{1\}$ 
not in $\Omega$. Let $\Lambda = \Omega \cap C_1^{c}$. As we are 
assuming that $C_1\setminus\{1\} \not\subset \Omega$, we have that 
$|\Lambda| \geq n_0/9$. Moreover, by definition of $i^{*}$, we have 
that $|\langle \ell_{1},\ell_{i^{*}}\rangle | \leq |\langle 
\ell_{1},\ell_{j}\rangle | $ for all $j \in \Omega$, and in particular:
\begin{equation*}
 |\langle \ell_{1},\ell_{i^{*}}\rangle | \leq |\langle \ell_{1},\ell_{j}\rangle | \text{ for all } j \in\Lambda 
\end{equation*}
Summing over $\Lambda$ we get:
\begin{equation}
 |\Lambda ||\langle \ell_{1},\ell_{i^{*}}\rangle | \leq \sum_{j\in\Lambda}|\langle \ell_{1},\ell_{j}\rangle | = \|\bfL_{\Lambda}^\top \ell_{1}\|_{1}
 \label{eq:ForceContradiction}
\end{equation} 
We shall show that equation \eqref{eq:ForceContradiction} cannot hold for $n_0$ large enough. Because $p \geq 4k(\ln(n))^2/n$ and $r << 1$, we may use Theorem \ref{thm:EBoundsMain}, which we shall do repeatedly. By this theorem, we have that $\bfL = \bfL^{0} + \bfE =
\bfL^{0}+ 
\bfE^{1} + \bfE^{2}$ with $\|\bfE^{j}\|_{\infty} \leq r = \mathcal{O}(r^2)$ 
for $j=1, 2$.  Moreover by construction $\bfL^{0}_{\Lambda} = 
\left[\begin{array}{c} \mathbf{0} \\
\tilde{\bfL}_{\Lambda} \end{array}\right]$ where $\tilde{\bfL}$ denotes 
the Laplacian of the subgraph $G_{C_1^{c}}$ and $\mathbf{0}$ is the zero matrix of size $n_0\times |\Lambda|$. Similarly, $\ell_{1} = \ell_1^{0} + \bfe_1$  and we may write $\ell_{1}^{0} = 
\left[\begin{matrix} \ell_{1}^{1} \\ 
\mathbf{0}\end{matrix}\right]$ where  
$\ell_{1}^{1}\in\mathbb{R}^{n_0}$ is the first column of 
$\bfL^{1}$, the Laplacian of the subgraph $G_{C_1}$ and $\mathbf{0}$ is the zero vector of length $n-n_0$. Thus:
\begin{align}
\|\bfL_{\Lambda}^\top \ell_{1}\|_{1} & \leq \|(\bfL^{0}_{\Lambda})^\top 
\ell^{0}_{1}\|_{1} + \|(\bfL^{0}_{\Lambda})^{\top}\bfe_{1}\|_{1} + 
\|\bfE^{\top}\ell^{1}\|_{1} \nonumber \\
&  = \left\|  \left[\begin{matrix} \mathbf{0}^{\top} & \tilde{\bfL}_{\Lambda}^\top  \end{matrix}\right] \left[\begin{matrix} 
\ell_{1}^{1} \\ \mathbf{0}\end{matrix}\right]\right\| 
+ \|(\bfL^{0}_{\Lambda})^\top \bfe_{1}\|_{1} + \|\bfE^\top \ell^{1}\|_{1} \nonumber \\
& \leq 0 + \|(\bfL^{0}_{\Lambda})^\top \|_{1}\|\bfe_{1}\|_{1} + \|\bfE^\top \|_{1}\|\ell^{1}\|_{1} \label{eq:ReturnToThis}
\end{align}
Now $\|(\bfL^{0}_{\Lambda})^\top \|_{1} = \|\bfL^{0}_{\Lambda}\|_{\infty}$, as these are dual norms. As $\|\cdot\|_{\infty}$ is equal to the maximum absolute row sum, it is clear that $\|\bfL^{0}_{\Lambda}\|_{\infty} \leq \|\bfL^{0}\|_{\infty}$. Now:
\begin{align*}
\|\bfL^{0}\|_{\infty} &= \max_{i}\sum_{j}\left| \delta_{ij} - \frac{\bfA^{0}_{ij}}{d^{0}_{i}}\right| = \max_{i}\left(1 + \frac{1}{d^{0}_{i}}\sum_{j\neq i}\bfA^{0}_{ij}\right) = \max_{i}\left(1 + \frac{1}{d^{0}_{i}}(d^{0}_{i})\right) = 2\\
\end{align*}
Moreover, $\|\bfe_1\|_{1} \leq \max_{i}\|\bfe_{i}\|_{1} = \|\bfE\|_{1}$, hence $\|\bfe_{1}\|_{1} \leq 2\beta^2r + \mathcal{O}(r^2)$, while $\|\bfE^{\top}\|_{1} = \|\bfE\|_{\infty} \leq \|\bfE^{1}\|_{\infty} + \|\bfE^{2}\|_{\infty} = 2r + \mathcal{O}(r^2)$ both by Theorem \ref{thm:EBoundsMain}. Finally, $\ell_{1} = \ell_{1}^{0} + \bfe_1$ as above, so $\|\ell_{1}\| \leq \|\ell_{1}^{0}\|_{1} + \beta^{2}r + \mathcal{O}(r^2)$. We compute $\|\ell_{1}^{0}\|_{1}$:
\begin{equation*}
\|\ell_{1}^{0}\|_{1} = \sum_{i=1}^{n}\left| \delta_{i1} - \frac{A^{0}_{i1}}{d^{0}_{i}}\right| = 1 + \sum_{i=2}^{n}\frac{\bfA^{0}_{i1}}{d^{0}_{i}} \leq 1 + \frac{1}{d_{\min}}(d^{0}_{1}) \leq 1 + \frac{(1+\alpha)pn_0}{(1-\alpha)pn_0} = 1 + \beta^{2}
\end{equation*}
Returning to equation \eqref{eq:ReturnToThis}:
\begin{equation*}
\|\bfL^{\top}_{\Lambda}\ell_1\|_{1} \leq (2)(2\beta^{2}r) + (2r)(1+ \beta^2) + \mathcal{O}(r^2) = (2 + 6\beta^2)r+ \mathcal{O}(r^2) = \frac{(2 + 6\beta^2)r_0}{\sqrt{n_0}}+ \mathcal{O}(r^2)
\end{equation*}

On the other hand, by Theorem \ref{thm:NoisyInnerProdBound} we have that:
\begin{equation*}
|\langle \ell_{1},\ell_{i^{*}}\rangle| \geq \frac{1}{n_0}\left(\frac{\beta^2}{1+\alpha}\right) - o(\frac{1}{n_0})
\end{equation*}
almost surely, and so we can bound the left hand side of 
\eqref{eq:ForceContradiction} as:
\begin{equation*}
|\Lambda |  |\langle \ell_{1},\ell_{i^{*}}\rangle| \geq \frac{n_0}{9}\left(\frac{1}{n_0}
\left(\frac{\beta^2}{1+\alpha}\right) - o(\frac{1}{n_0}) 
\right) = \frac{\beta^2}{9(1+\alpha)} - o(1)
\end{equation*}
Thus, if inequality \eqref{eq:ForceContradiction} were true, it would imply that:
\begin{equation*}
\frac{\beta^2}{9(1+\alpha)} - o(1) \leq \frac{(2 + 6\beta^2)r_0}{\sqrt{n_0}} + \mathcal{O}(r^2)
\end{equation*}
which cannot be true for $n_0$ large enough, as $\alpha$ and $r$ are $o(1)$ with 
respect to $n_0$ and $\beta^2$ is $O(1)$, thus $\beta^2/(9(1+\alpha)) - 
o(1) \to 1/9$  while $(2 + 6\beta^2)r_0/\sqrt{n_0} \to 0$. Hence we may always take 
$n_0$ large enough so that, with probability at least $1 - 1/n_0$, 
inequality \eqref{eq:ForceContradiction} cannot hold. Thus no such 
$i^{*}$ which is in $C_1\setminus\{1\}$ but not in $\Omega$ can 
be found, and so the Theorem~\ref{thm:Mckenzie}  is proved.
\end{proof}

With the result of Theorem~\ref{thm:Mckenzie} in mind, 
we are able to derive a new algorithm to find the communities 
from $G\in {\cal G}(n,k,p,q)$ for $p>>q>0$:
\begin{algorithm}[H]
\caption{Single Cluster Pursuit (SCP)}
\label{algorithm:CompClust}
Input:  The adjacency matrix  $A$ of a graph $G$, and estimate of the 
size of clusters $n_0$
\begin{algorithmic}
\State (1) {\bf Initialization } Compute $L = I - D^{-1}A$.
\State (2) {\bf Trimming} 
	\Indent
	\State Let $\Omega = \mathcal{L}_{\lceil 10(n_0-1)/9 \rceil}(L_{-1}^\top \ell_{1})$
	\EndIndent
\State (3) {\bf Perturbed Sparse Recovery}
	\Indent
    \State (a) $\bfy = \sum_{i\in\Omega}\ell_{i} + \ell_{1}$
	\State (b) Solve $\bfz^{\#} = \argmin \{\|L_{\Omega}\bfz - \bfy\|_{2} \ \text{ s.t. } \bfz\in \mathbb{R}^{|\Omega|} \text{ and } \|\bfz\|_{0} \leq n_0/9\}$ using the SP algorithm
	\State (c) $\Lambda^{\#} = \supp(\bfz^{\#})$.
	\EndIndent
\end{algorithmic}
Output: $C_{1}^{\#} = \{1\}\cup \left(\Omega\setminus\Lambda^{\#}\right)$.
\end{algorithm}

Next we show that Algorithm~\ref{algorithm:CompClust} succeeds almost surely.

\begin{theorem}
\label{thm:MainSuccess}
Let $G$ be a graph drawn at random from the $\mathcal{G}(n,k,p,q)$ model. Suppose that $k = O(1)$ with respect to $n$ and either:
\begin{enumerate}
\item $q = \frac{Q}{n}$ and $p = \frac{P\ln(n)}{\sqrt{n}\ln(\ln(n))}$ with $Q$ and $P$ being constants and $P > 979$, or
\item $q = \frac{Q\ln(n)}{n}$ with $Q\to \infty$ as $n\to\infty$ and $p = \frac{P\ln(n)}{\sqrt{n}}$ with $P\to\infty$ as $n\to\infty$ and $\frac{Q}{P} < \frac{1}{979}$ for large enough $n$.
\end{enumerate}
Then Algorithm~\ref{algorithm:CompClust} will recover $C_1$ (that is, $C_1^{\#} = C_1$) almost surely. 
\end{theorem}
\begin{proof}
For notational convenience let $s:= (n_0-1)/9$. In both cases, $p \geq 4k(\ln(n))^2/n$ and so by part 2 of Theorem \ref{lemma:AlphaDegBound}, $d^{0}_{\min} \geq (1-\alpha)n_0p$ almost surely for $\alpha = o(1)$ with respect to n. In case 1, by (i) of Theorem  3.4 in \cite{Frieze2016}, $d^{\epsilon}_{\max} \leq \ln(n)/(\ln(\ln(n)) - \ln(\ln(\ln(n)))$ almost surely. Ignoring the triple-logarithmic term:
\begin{equation*}
\max_{i}\frac{d^{\epsilon}_{i}}{d^{0}_{i}} \leq \frac{d^{\epsilon}_{\max}}{d^{0}_{\min}} \leq \frac{1}{P(1-\alpha)\sqrt{n_0}} =: \frac{r_0}{\sqrt{n_0}}
\end{equation*}
with $r_0 < 1/979$ for large enough $n$ as $\alpha = o(1)$. In case $2$, by (ii) of Theorem  3.4 in \cite{Frieze2016}, $d^{\epsilon}_{\max}\leq Q\ln(n)$ almost surely, and so:
\begin{equation*}
\max_{i}\frac{d^{\epsilon}_{i}}{d^{0}_{i}} \leq \frac{d^{\epsilon}_{\max}}{d^{0}_{\min}} \leq \frac{Q/P}{(1-\alpha)\sqrt{n_0}} =: \frac{r_0}{\sqrt{n_0}}
\end{equation*}
and again $r_0 < 1/979$ for large enough $n$. We may thus apply Theorem~\ref{thm:Mckenzie} to get that $C_{1}\setminus\{1\}\subset \Omega$ almost surely. As before we define $\Lambda := \Omega\cap 
C_{1}^{c}$; it is easy to see $|\Lambda|\leq n_0/9$. The challenge in step (3) of Algorithm~\ref{algorithm:CompClust} is to separate $C_{1}\setminus\{1\}$ from $\Lambda$. We do this by solving for $\bfone_{\Lambda} = \bfone_{\Omega} - \bfone_{C_{1}\setminus\{1\}}$ instead of $\bfone_{C_{1}\setminus\{1\}}$, as:
\begin{equation}
\bfL^{0}_{\Omega}\bfone_{\Lambda} = \bfL^{0}_{\Omega}\left(\bfone_{\Omega} - \bfone_{C_{1}\setminus\{1\}}\right) = \bfL^{0}_{\Omega}\bfone_{\Omega} - \bfL^{0}_{\Omega}\bfone_{C_{1}\setminus\{1\}} = \sum_{i\in \Omega}\ell^{0}_{i} - (-\ell^{0}_1) =: \bfy^{0}
\label{eq:switch}
\end{equation}
and noting that $\|\bfone_{\Lambda}\|_{0} \leq n_0/9$, we see that  $\bfone_{\Lambda}$ is the unique solution to 
\begin{equation}
\argmin \{\|\bfL^{0}_{\Omega}\bfz - \bfy^{0}\|_{2}:\ \bfz \in \mathbb{R}^{\lceil 10n_0/9\rceil}, \ \|\bfz\|_{0} \leq n_0/9\}
\label{eq:unpeturbed}
\end{equation}
Defining $\bfy$ to be a perturbation of $\bfy^{0}$and writing, as in \S \ref{sectionRIP_q_nonzero}, $\bfL_{\Omega} = \bfL^{0}_{\Omega} + \bfE_{\Omega}$:
\begin{equation*}
\bfy := \sum_{i\in\Omega}\ell_{i} + \ell_1 = \sum_{i\in \Omega}\ell^{0}_{i} + \sum_{i\in\Omega}\bfe_{i} + \ell_{1}^{0} + \bfe_1 = \bfy^{0} + \left(\bfe_1 + \sum_{i\in\Omega}\bfe_{i}\right)
\end{equation*}
Defining $\bfe = \bfe_1 + \sum_{i\in\Omega}\bfe_{i}$ we recognize the problem:
\begin{equation}
\label{eq:Perturbed2}
\bfz^{\#} = \argmin \{\|\bfL_{\Omega}\bfz - \bfy \|_{2}:\ \bfz \in \mathbb{R}^{\lceil 10n_0/9\rceil}, \ \|\bfz\|_{0} \leq n_0/9\}
\end{equation}
as a totally perturbed version of \eqref{eq:unpeturbed}. Thus, we may apply the results of \S \ref{section:PerturbedCS} to bound $\|\bfz^{\#} - \bfone_{\Lambda}\|_{2}$. As in \S \ref{section:PerturbedCS}, we define $\delta_{t} := 
\delta_{t}(\bfL^{0})$ and $\hat{\delta}_{t}:= \delta_{t}(\bfL)$. 
(We refer the reader to \S \ref{section:PerturbedCS} for 
the definitions of $\epsilon^{s}_{\bfL^{0}}$, $\epsilon_{\bfy}$, $\hat{\rho}$ and $\hat{\tau}$.)  
Let us bound on these quantities.  Observe that in both cases, 
$p >> (\ln(n_0))^2/n_0$ and by assumption $k$ is $O(1)$ with respect to $n$,  thus as in the proof of Theorem \ref{thm:Twostory} we get
$\epsilon^{s}_{\bfL^{0}} = 8\beta r/3 = 8\beta r_0/(3\sqrt{n_0})$, for $n_0$ large enough. 
Applying the result of Theorem \ref{thm:Twostory} with $t = 3s = n_0/3$:
 \begin{align*}
 \hat{\delta}_{3s} &\leq \frac{1}{3} + D_1r + \frac{D_2}{\sqrt{n_0}} + \frac{D_3r}{\sqrt{n_0}} + \mathcal{O}(r^2) + 
o(\frac{1}{\sqrt{n_0}}) \\
& = \frac{1}{3} + \frac{(D_1r_0+D_2)}{\sqrt{n_0}} + \frac{D_3r_0}{n_0} + \mathcal{O}(\frac{1}{n_0}) + 
o(\frac{1}{\sqrt{n_0}}) = \frac{1}{3} + \frac{(D_1r_0+D_2)}{\sqrt{n_0}} + o(\frac{1}{\sqrt{n_0}}).
\end{align*}
Hence for large enough $n_0$, we certainly have $\hat{\delta}_{3s} \leq 0.4859$ as required. In fact, let us take 
$n_0$ large enough such that $\hat{\delta}_{3s} \leq 1/3 + 1/18 = 7/18 \approx 0.39$. Choosing $n_0$ larger if 
necessary, we shall also assume that $\epsilon^{s}_{L^{0}} < 1/3$. In this case, a straightforward, but tedious 
calculation will reveal that $\hat{\rho} \leq 0.7$ and similarly $\hat{\tau} < 166$ (see the statement of 
Theorem \ref{thm:PerturbedSP} for the definitions of $\hat{\rho}$ and $\hat{\tau}$). We now turn our attention to 
$\epsilon_{\bfy}$.
\begin{equation*}
\|\bfy^{0}\|_{2} = \|\bfL^{0}_{\Omega_1}\bfone_{\Lambda}\|_{2} \geq \sqrt{1 - \delta_{s}}\|\bfone_{\Lambda}\|_{2} = 
\sqrt{1 - \delta_{s}}\sqrt{s}
\end{equation*}
while 
\begin{equation*}
\|\bfe\|_{2}  = \left\| \bfE\bfone_{\Omega} +\bfe_1 \right\|_{2} \leq \|\bfE\|_{2}\|\bfone_{\Omega\cup\{1\}}\|_{2} \leq \left(2
\beta r + \mathcal{O}(r^2)\right) \sqrt{10s}.
\end{equation*}   
So we have 
\begin{equation*}
\epsilon_{\bfy} := \frac{\|\bfe\|_{2}}{\|\bfy^{0}\|_{2}} = \frac{2\sqrt{10}\left(2
\beta r + \mathcal{O}(r^2)\right)}{\sqrt{1-\delta_{s}}} = \frac{2
\sqrt{10}\beta r_0}{\sqrt{n_0(1-\delta_s)}} + o(\frac{1}{\sqrt{n_0}})
\end{equation*}
Appealing to Theorem \ref{thm:Twostory} with $s = n_0/9$:
\begin{equation*}
\delta_{s} \leq 1/9 + \frac{64p^{-1/2}}{9\sqrt{n_0}} + o(\frac{1}{\sqrt{n_0}}). 
\end{equation*}
Again, we will assume that $n_0$ is large enough such that $\delta_{s} \leq 2/9 \approx 0.22$. 
Under this assumption:
\begin{equation*}
\epsilon_{\bfy} \leq \frac{2\sqrt{10}}{\sqrt{7/9}}\frac{\beta r_0}{\sqrt{n_0}} \leq 8 \frac{\beta r_0}{\sqrt{n_0}}
\end{equation*}
Finally, we use Theorem \ref{thm:Twostory} again, this time with $t = s = n_0/9$:
 \begin{equation*}
 \hat{\delta}_{s} \leq  1/9 +\frac{D_1r_0+D_2}{\sqrt{n_0}} + o(\frac{1}{\sqrt{n_0}})
 \end{equation*}
 and take $n_0$ large enough such that $\hat{\delta}_{s} \leq 2/9$ (this is possible as $D_1$ and $D_2$ are 
$\mathcal{O}(1)$ with respect to $n_0$). Now Theorem \ref{thm:PerturbedSP} guarantees that 
 \begin{equation*}
 \frac{\|\bfx^{m} - \bfone_{\Lambda}\|_{2}}{\|\bfone_{\Lambda}\|_{2}} \leq \left( \hat{\tau}\frac{\sqrt{1 + \hat{
\delta}_{s}}}{1 - \epsilon^{s}_{L^{0}}} + 1 \right)(\epsilon_{L^0}^{s} + \epsilon_{\bfy})
 \end{equation*}
 Where $m := \ln(\epsilon_{\bfL^0}^{s} + \epsilon_{\bfy})/\ln(\hat{\rho}) = \mathcal{O}(\ln(n_0))$. Substituting in for the various constants:
 \begin{align*}
& \frac{\|\bfx^{m} - \bfone_{\Lambda}\|_{2}}{\|\bfone_{\Lambda}\|_{2}} < \left( (166)\frac{\sqrt{1 + 2/9}}{1 - 1/3} 
\right)\left(\frac{8\beta r_0}{3\sqrt{n_0}} + \frac{8\beta r_0}{\sqrt{n_0}} \right) = \frac{166(\sqrt{11})(32)}{6}\frac{\beta r_0}{\sqrt{n_0}} \\
\Rightarrow & \|\bfx^{m} - \bfone_{\Lambda}\|_{2} < \left(\frac{\sqrt{n_0}}{3}\right) 2937 \frac{\beta r_0}{\sqrt{n_0}} = 979\beta r_0
 \end{align*}
As the Lemma~\ref{cp} below, as long as $\|\bfx^{m} - \bfone_{\Lambda}\|_{2} < 1$, we have that 
$\supp(\bfx^{m}) = \supp(\bfone_{\Lambda})$. Because $\beta \to 1$ and $r_0 < 1/979$, this will hold for large enough $n_0$.
\end{proof}

\begin{lemma}
\label{cp}
Suppose that $\bfx^{*}\in\mathbb{R}^{n}$ is a binary vector with $\|\bfx^{*}\|_{0} = t$, and $\bfx\in\mathbb{R}^{n}$ is any other vector that also has $\|\bfx\|_{0} = t$. If $\|\bfx^{*} - \bfx\|_{2} < 1$ then $\supp(\bfx) = \supp(\bfx^{*})$.
\end{lemma}
\begin{proof}
Suppose otherwise. Then there exists an $i\in \supp(\bfx^{*})\setminus \supp(\bfx)$. Clearly:
\begin{equation*}
\|\bfx^{*} - \bfx\|_{2} \geq |x^{*}_{i} - x_{i}| = |1 - 0| = 1
\end{equation*}
which contradicts the hypotheses.
\end{proof}

As mentioned in the introduction, we may iterate Algorithm \ref{algorithm:CompClust} to find all the clusters of $G$. We call this algorithm Iterated Single Cluster Pursuit, or ISCP.

\section{Computational Complexity and Extensions}
 In this section we bound the computational complexity and explain some extensions.

\subsection{Computational Complexity}
 In this subsection we show that Algorithm~\ref{algorithm:CompClust} 
is faster than existing spectral methods. We do this by determining the 
approximate number of operations required in each step of 
Algorithm \ref{algorithm:CompClust}. Throughout, we shall assume that $\bfA$ and $\bfL$ are sparse matrices.
 \begin{enumerate}
\item Computing $d_{i} = \sum_{j}A_{ij}$ requires $d_i$ operations. We can safely assume that $d_{i} = d^{0}_{i} + d^{\epsilon}_{i} \leq n_0$ so the cost of computing each $d_i$ is at most $\mathcal{O}(n_0)$ This needs to be done a total of $n$ times to compute $\mathbf{D} = \text{diag}(d_1,\ldots, d_n)$, giving a total cost of $\mathcal{O}(nn_0)$.
 \item The cost of the thresholding step is dominated by the matrix-vector multiply $\bfL^\top _{-1}\ell_1$. Because $\ell_1$ has $d_1 + 1 \leq n_0$ non-zero entries, the cost of this is bounded by $\mathcal{O}(nn_0)$.
 \item The computational cost of solving the perturbed sparse recovery problem in step 3 using SP is equal to the number of iteration times the cost of each iteration. As shown in the proof of 
Theorem~\ref{thm:MainSuccess}, it suffices to perform 
$\mathcal{O}(\ln(n_0)) = \mathcal{O}(\ln(n))$ iterations. The cost of each 
iteration is determined by calculating the cost of each step in the 
iterative part of SP (see Algorithm \ref{algorithm:SP}): 
\begin{enumerate}
 \item Computing $\mathcal{L}_{s}(\bfL^\top _{\Omega}\bfr^{k-1})$ is dominated by the cost of the matrix-vector multiply $\bfL^\top _{\Omega}\bfr^{k-1}$, hence is $\mathcal{O}(nn_0)$.
 \item Solving the least square problem in step (2) using a good numerical solver (we used {\tt MATLAB}'s {\tt lsqr } algorithm) 
is on the order of the cost of a matrix-vector multiply, i.e. 
$\mathcal{O}(nn_0)$,  as explained in \cite{Needell2009}.
\item The cost of sorting and thresholding (step (3)) is $\mathcal{O}(n\ln(n))$.
 \item Finally the cost of computing the new residual $\bfr^{k}$ in step (4) is dominated by the matrix vector multiply $\bfL^\top _{\Omega_1}\bfr^{k}$, hence is $\mathcal{O}(nn_0)$.
 \end{enumerate} 
We conclude that the cost of a single iteration of  the SP Algorithm  is $\mathcal{O}(nn_0)$, and so the cost of step 3 of Algorithm~\ref{algorithm:CompClust} is $\mathcal{O}(n\ln(n)n_0)$
\end{enumerate}
Thus, the number of computations required to find a single cluster using Algorithm~\ref{algorithm:CompClust} is \mbox{$\mathcal{O}(n\ln(n)n_0)$}.  \\ 

To find all the clusters $C_{1},\ldots, C_{k}$ one can iterate Algorithm~\ref{algorithm:CompClust} $k-1$ times. The cost of this is 
certainly less than $\mathcal{O}(kn\ln(n)n_0) = \mathcal{O}(n^2\ln(n))$.


\subsection{Extensions}
\label{section:Extensions}
\begin{enumerate}
\item Our Algorithm~\ref{algorithm:CompClust} 
can be extended to deal with the communities of
non-equal size. Hence, our study can be extended too. 
A numerical example is included 
in the next section to demonstrate this extension. For simplicity, 
we leave the study for the interested reader.
\item Although Algorithm~\ref{algorithm:CompClust} 
works on binary adjacency matrices, it handles weighted adjacency matrices just as well (see \S \ref{section:GeneExpression}).
\item As mentioned in the introduction, graph clustering can also be 
applied to more general data sets. Given any (finite) set of points $\bfX 
= \{\bfx_1,\ldots, \bfx_{n}\}$ in some metric space $(M,d)$, there are 
several ways to associate a graph on $n$ vertices to $\bfX$. For example, 
we could attach an edge between vertices $i$ and $j$ whenever 
$d(\bfx_i,\bfx_j) \leq C$ for some constant $C$. Alternatively, for each vertex $i$ we could 
determine the $K$ closest points to $\bfx_i$, say $\{\bfx_{j_1}, 
\ldots, \bfx_{j_K}\}$, and insert the edges $\{i,j_1\},\ldots, 
\{i,j_K\}$ to obtain the \emph{$K$-nearest neighbours graph} $G^{K}$. 
Empirically, we have found the latter approach to perform better, as 
$d_{\min}(G^{K}) = K$ giving better control over the quantities $r_i$ 
crucial to the analysis of Algorithm~\ref{algorithm:CompClust}. See \S \ref{section:GeneExpression} for an example of this approach.
\item Algorithm \ref{algorithm:CompClust} can easily be extended to the co-clustering problem, as described in \cite{Dhillon2001}. Briefly suppose that we are given two categorically different sets of variables $\mathbf{X} = \{\bfx_1,\ldots, \bfx_n \}$ and $\mathbf{Y} = \{\bfy_1,\ldots, \bfy_{m}\}$ as well as a $n\times m$ matrix $\mathbf{B}$ such that the $\mathbf{B}_{ij}$-th entry represents the correlation between $\bfx_{i}$ and $\bfy_{j}$. The goal of co-clustering is to simultaneously partition $\mathbf{X}$ and $\mathbf{Y}$ into subsets $\mathbf{X} = \mathbf{X}_1\cup\ldots \mathbf{X}_{k}$ and $\mathbf{Y} = \mathbf{Y}_{1}\cup\ldots \cup \mathbf{Y}_{k}$ such that the correlations between $\bfx \in \mathbf{X}_{i}$ and $\bfy \in \mathbf{Y}_{i}$ are strong while the correlations between $\bfx \in \mathbf{X}_{i}$ and $\bfy \in \mathbf{Y}_{j}$ for $i\neq j$ are weak. As pointed out by Dhillon in \cite{Dhillon2001}, one can regard this as a graph clustering problem for a weighted\footnote{We assume that the correlations $B_{ij}$ are non-negative} bipartite graph with vertex set $V = \{u_{1},\ldots, u_{n}, v_{1},\ldots, v_{m}\}$ and weights $w(u_{i},v_{j}) = B_{ij},\  w(u_i,u_j) = w(v_i,v_j) = 0$. The adjacency matrix and Laplacian of G are:
\begin{equation*}
\bfA = \left[\begin{matrix} \mathbf{0} & \mathbf{B} \\ \mathbf{B}^{\top} & \mathbf{0} \end{matrix}\right] \quad \quad \bfL = \left[\begin{matrix} \mathbf{I}_{n\times n} & -\mathbf{D}_{\bfX}^{-1}\mathbf{B} \\ -\mathbf{D}_{\bfY}^{-1}\mathbf{B}^{\top} & \mathbf{I}_{m\times m}\end{matrix}\right]
\end{equation*}
where $\mathbf{D}_{\bfX}$ represents the row sums of $\mathbf{B}$ while $\mathbf{D}_{\bfY}$ represents the column sums. Assume temporarily that instead of clusters $G$ has $k$ connected components $C_{1},\ldots, C_{k}$. Each $G_{C_i}$ will be a bipartite graph, and $C_{i}$ splits as $C_{i} = C_{i}^{\bfX} \cup C_{i}^{\bfY}$ where $C_{i}^{\bfX}$ corresponds to $\bfX_{i}$ and similarly $C_{i}^{\bfY}$ corresponds to $\bfY_i$. The indicator vector will split similarly: $\bfone_{C_i} = [\bfone_{\bfX_i}^{\top} \ \bfone_{\bfY_i}^{\top} ]^{\top}$ and:
\begin{equation*}
\mathbf{0} = \bfL\bfone_{C_i} = \left[\begin{matrix} \bfone_{\bfX_i} -\mathbf{D}_{\bfX}^{-1}\mathbf{B}\bfone_{\bfY_i} \\  \bfone_{\bfY_i} -\mathbf{D}_{\bfY}^{-1}\mathbf{B}^{\top}\bfone_{\bfX_i} \end{matrix}\right]
\end{equation*}
We may convert these two `first-order' equations into a single `second-order' equation as $\bfone_{\bfY_i} = \mathbf{D}_{\bfY}^{-1}\mathbf{B}^{\top}\bfone_{\bfX_i}$ and so:
\begin{equation*}
\mathbf{0} = \bfone_{\bfX_i}  - \mathbf{D}_{\bfX}^{-1}\mathbf{B}\mathbf{D}_{\bfY}^{-1}\mathbf{B}^{\top}\bfone_{\bfX_i} = \left(\mathbf{I} - \mathbf{D}_{\bfX}^{-1}\mathbf{B}\mathbf{D}_{\bfY}^{-1}\mathbf{B}^{\top}\right)\bfone_{\bfX_i}
\end{equation*}
Defining the `bi-partite Laplacian' as $\bfL^{BP} = \mathbf{I} - \mathbf{D}_{\bfX}^{-1}\mathbf{B}\mathbf{D}_{\bfY}^{-1}\mathbf{B}^{\top}$ we see that if $\bfx_1 \in \bfX_1$ then the indicator vector $\bfone_{\bfX_1}$ is the unique solution to:
\begin{equation*}
\argmin \{ \|\bfz\|_{0}: \quad \bfz \in \mathbb{R}^{n}, \ \bfL^{BP}\bfz = \mathbf{0} \text{ and } z_{1} = 1\}
\end{equation*}
once we have found $\bfone_{\bfX_1}$ we can recover $\bfone_{\bfY_{1}}$ as $\bfone_{\bfY_1} = D_{\bfY}^{-1}B^{\top}\bfone_{\bfX_1}$. Returning to the clustering problem, we see that we may approximate $\bfone_{\bfX_1}$ ( and hence $\bfone_{\bfY_i}$) by solving the problem 
\begin{equation*}
\argmin \{ \|\bfL^{BP}_{-1}\bfz - (-\ell_{BP,1})\|_{2}:\quad \bfz \in \mathbb{R}^{n-1}, \ \|\bfz\|_{0} \leq n^{\bfX}_{0} - 1\}
\end{equation*}
using Algorithm \ref{algorithm:CompClust}, where $n^{\bfX}_{0}$ is the size of the clusters $\bfX_i$. We leave the technical analysis of this approach to the interested reader.
\item In future work, we intend to show how Algorithm \ref{algorithm:CompClust} can be adapted into an `online' algorithm which allows one to rapidly update the community assignments as more vertices are added to the graph.
\end{enumerate}


\section{Numerical Examples}
\subsection{Synthetic Data}
We tested Single Cluster Pursuit (SCP, Algorithm \ref{algorithm:CompClust}) on graphs drawn from the SBM. We implemented this algorithm in {\tt MATLAB} 
2016b and run on a mid-2010 iMac computer with 8 GB of RAM and a $3.06$ 
GHz Intel Core i3 processor. For comparison we used the Trembley, Puy, 
Gribonval and Vandergheynst implementation of Spectral Clustering (SC,
Algorithm \ref{algorithm:SC}) included in their Compressive Spectral 
Clustering toolbox available at \url{http://cscbox.gforge.inria.fr/} 
(\cite{Tremblay2016}).

\subsubsection{Example 1}
 Typical output of SCP applied to randomly drawn graphs from the SMB $\mathcal{G}(n,k,p,q)$ with $p = P\ln(n)/\sqrt{n}$ and $q= Q\ln(n)/n$ as in case $2$ of Theorem \ref{thm:MainSuccess} is shown in the third frame of figure \ref{fig:OneCluster}.  Both SCP and SC succeed in finding cluster 1 without error (indeed SC finds all clusters without error) but SCP is faster, taking $0.0262$ seconds compared to the $0.1018$ seconds required by SC. We remark that even for $Q$, $P$ such that $Q/P >> 1/979$ SCP is successful, suggesting that the bounds in Theorem \ref{thm:MainSuccess} are too conservative.

\begin{figure}[H]
\begin{center}
\bigskip
\includegraphics[width = 6in, height=4in]{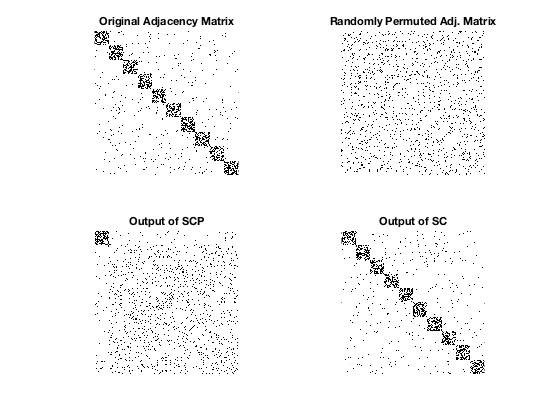}
\end{center}
\caption{The adjacency matrix of a graph drawn from $\mathcal{G}(1000,5,\frac{2\ln(n)}{\sqrt{n}},\frac{2\ln(n)}{n})$ is shown in the top left. The same matrix, with its rows and columns randomly permuted (using the same permutation for both) is shown in the top right. The output of SCP applied to the permuted matrix is shown in the bottom left, while the output of SC on the same matrix is shown on the bottom right.}\label{fig:OneCluster}
\end{figure}

\subsubsection{Example 2}
We then tested Iterated SCP (ISCP) for graphs randomly drawn from the SBM  $\mathcal{G}(n,k,p,q)$ with $p = P\ln(n)/\sqrt{n}$ and $q= Q\ln(n)/n$. Typical output is shown in Figure~\ref{fig:full_compare}. Again, both algorithms are successful (that is, both find all $10$ communities without error). 

\begin{figure}[H]
\begin{center}
\bigskip
\includegraphics[width = 6in, height=4in]{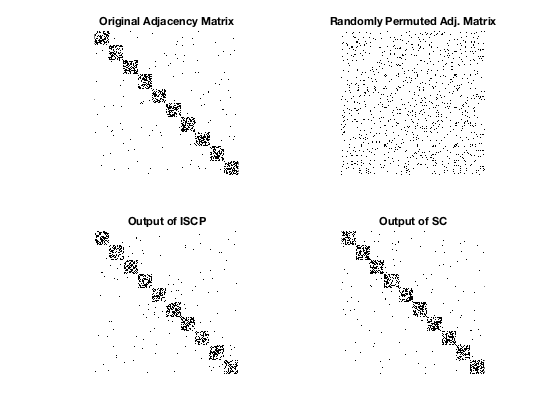}
\end{center}
\caption{The adjacency matrix of a graph drawn at random from $\mathcal{G}(5000,10,\frac{4\ln(n)}{\sqrt{n}},\frac{4\ln(n)}{n})$ is shown in the top left. The same matrix, with its rows and columns randomly permuted (using the same permutation for both) is shown in the top right. The output of SCP applied to the permuted matrix is shown in the bottom left, while the output of SC on the same matrix is shown on the bottom right.}\label{fig:full_compare}
\end{figure}

\subsubsection{Example 3}
Next, we tested the resilience of SCP to noise by running it on $G$ drawn at random from $\mathcal{G}(2400,6,0.5,q)$ with $q$ increasing from $0$ to $150/2000$. For each $q$, we ran SCP on $10$ different $G$ drawn independently and at random from $\mathcal{G}(2400,6,  
0.5,q)$ and computed the average fraction of indices that were misclassified. Figure \ref{fig:ResTest} shows the result of this experiment, with $\mathbb{E}[d^{\epsilon}_{i}] = q(n-n_0)$ on the $x$-axis. As observed in Experiment 1, SCP performs much better than theoretically guaranteed by Theorem \ref{thm:MainSuccess}. Indeed, Figure \ref{fig:ResTest} demonstrates that SCP detects community 1 without error up until $\mathbb{E}[d^{\epsilon}_{i}] = 40$, at which $r \approx 40/200 = 1/5$.  
\begin{figure}[H]
\begin{center}
\includegraphics[width = 3in, height=2in]{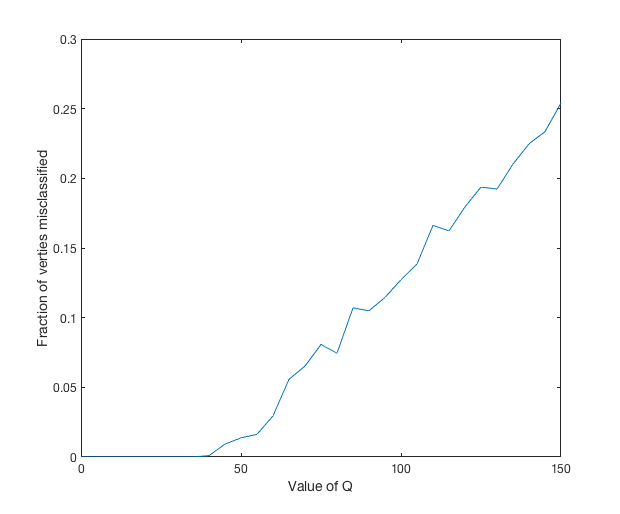}
\end{center}
\caption{Fraction of indices misclassified by SCP, for $G \in \mathcal{G}(2400,6,0.5,Q/2000)$. The horizontal axis represents $Q$, that is, the level of `noise' in the graph. The vertical axis represents the fraction of misclassified vertices $ = |C_1^{\#}\setminus C_1|/|C_1^{\#}|$. As can be seen, SCP performs without error until $Q \approx 40$, well beyond the theoretical guarantees given by Theorem \ref{thm:MainSuccess}} 
\label{fig:ResTest}
\end{figure}
\subsubsection{Example 5}
In our final example, we use ISCP to solve the co-clustering problem, as discussed in \S \ref{section:Extensions}. Typical output is shown in image \ref{fig:RectangleCluster}

\begin{figure}[H]
\begin{center}
\bigskip
\includegraphics[width = 0.7\textwidth]{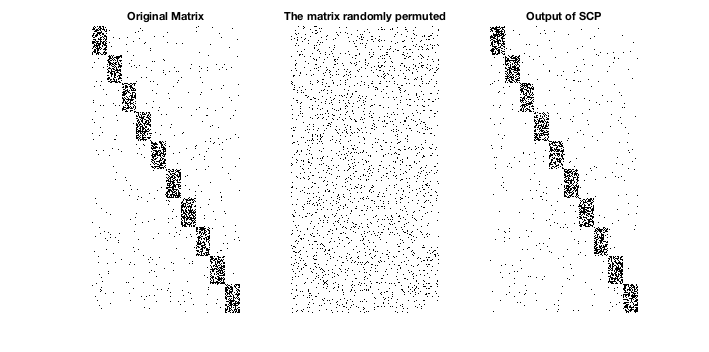}
\end{center}
\caption{A $2000\times 1000$ binary matrix with $10$ equally sized block, randomly permuted and then co-clustered using ISCP, applied as discussed in \S \ref{section:Extensions}.}\label{fig:RectangleCluster}
\end{figure}

\subsubsection{Example 4}
Finally we compared the run times of SCP, SC and ISCP on $\mathcal{G}(n,k,p,q)$ for increasing $n$. We studied three regimes: constant $k$, constant $n_0$ (cluster size) and $k,n_0 = \mathcal{O}(\sqrt{n})$. In all three cases we fixed $p = 2\log(n)/\sqrt{n}$ and $q = 2\log(n)/n$. The results of these experiments are presented in figures \ref{fig:CompareTime20Clusters} - \ref{fig:CompareTimevaryKAndN0} (the run times are in seconds, and are the average of ten independent trials for the same $n$). As is clear, SCP significantly outperforms SC in finding a single cluster. Moreover, when $k$ is large enough compared to the number of vertices ($n$), ISCP finds all clusters faster than SC.  

\begin{figure}[H]
\centering
\subfloat[Varying $k$ from $20$ to $120$ and $n$ accordingly, while fixing $n_0$. $G \in \mathcal{G}\left(400k,k,\frac{2\log(n)}{\sqrt{n}},\frac{2\log(n)}{n}\right)$]{\includegraphics{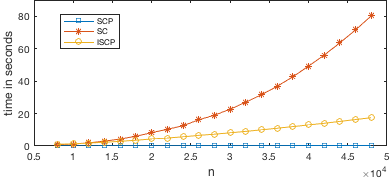}\label{fig:CompareTimevaryK}}

\subfloat[Varying $k$, $n_0$ and $n$. $G \in \mathcal{G}\left(kn_0,k, \frac{2\log(n)}{\sqrt{n}},\frac{2\log(n)}{n}\right)$]{\includegraphics{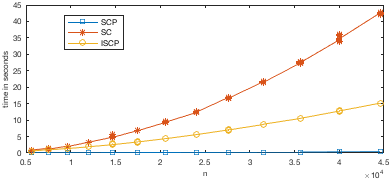}\label{fig:CompareTimevaryKAndN0}}

\subfloat[Varying $n$ and $n_0$ while fixing $k$. $G \in \mathcal{G}\left(n,40,\frac{2\log(n)}{\sqrt(n)}, \frac{2\log(n)}{n}\right)$]{\includegraphics{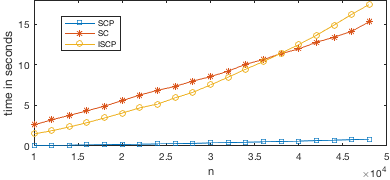}\label{fig:CompareTime20Clusters}}

\caption{Comparing SCP, ISCP and SC for graphs drawn from the SBM of increasing size. As can be seen, SCP is significantly faster than SC}
\end{figure}

\subsection{Real Data Sets}
We now present two examples of SCP and ISCP applied to real-world data sets.

\subsubsection{Political Blogs Data Set}
The {\tt polblogs} data set is a collection of $1494$ political weblogs, or blogs, collected by Adamic and Glance (\cite{Adamic2005}) prior to the 2004 American presidential election. The data is presented as an unweighted, undirected graph, with vertices corresponding to the blogs, and edges between blogs $i$ and $j$ if there is a hyperlink from $i$ to $j$, or \emph{vice versa}. This data set is well-studied (see, for example \cite{Krzakala2013}, \cite{Newman2006} and \cite{Olhede2014}, amongst others), and is a good test case for at least two reasons:
\begin{enumerate}
\item In addition to a natural division into two roughly equally sized clusters (liberal vs. conservative), the data set exhibits additional clustering at smaller scales. That is, within the set of all liberal blogs, one can identify subclusters of blogs, and similarly for the conservative blogs. This is explored further in \cite{Olhede2014}.
\item The ground truth for the division into liberal \emph{vs.} conservative is known, as Adamic and Glance (\cite{Adamic2005}) manually labelled the data set. This provides an opportunity to verify one's results that is rare for real-world data sets.
\end{enumerate}

Our methodology was as follows. We follow Olhede and Wolfe (\cite{Olhede2014}) in using only the $1224$ blogs with links to at least one other blog in the data set. As SCP can be thrown off by low degree vertices ( these will have high values of $r_i$) we experimented with different thresholds. That is, we discarded vertices of degree lower than $d_{\text{thresh}}$ for varying values of $d_{\text{thresh}}$. We then used SCP to detect a cluster containing an arbitrary liberal blog, of size approximately equal to the number of liberal blogs. We call this cluster Cluster 1 ($C_1$). We call the remaining vertices Cluster 2 $C_2$. The results are tabulated in \ref{table:polblogs}. We record the number of vertices left after thresholding, the percentage of the first cluster consisting of liberal blogs, the percentage of the second cluster consisting of conservative blogs and the time taken. The values shown are the averages of ten independent runs.

\begin{table}[H]
\begin{center}
\begin{tabular}{ccccc}
\hline
$d_{\text{thresh}}$ & $\#$ of vertices & $\%$ of $C_1$ liberal & $\%$ of $C_2$ conservative & time (in seconds) \\
\hline
$0$  & $1224$ & $63.71$ & $66.93$ & $0.0380$ \\
$2$  & $1087$ & $71.57$ & $75.68$ & $0.0928$ \\
$4$  & $903$  & $81.45$ & $83.86$ & $0.0467$ \\
$6$  & $813$  & $89.12$ & $91.06$ & $0.0479$ \\
$8$  & $735$  & $91.95$ & $93.69$ & $0.0428$ \\
$\mathbf{10}$ & $\mathbf{693}$  & $\mathbf{93.14}$ & $\mathbf{94.57}$ & $\mathbf{0.0411}$ \\
$12$ & $644$  & $93.95$ & $95.32$ & $0.0256$ \\
$14$ & $596$  & $93.05$ & $94.66$ & $0.0377$ \\
$16$ & $547$  & $94.24$ & $95.39$ & $0.0379$ \\
$18$ & $505$  & $93.97$ & $94.87$ & $0.0225$ \\
\hline
\end{tabular}
\label{table:polblogs}
\caption{Results of SCP applied to the {\tt polblogs} data set for varying degree thresholds. When $d_{\text{thresh}} = 10$ the fraction of cluster 1 classified as liberal, and the fraction of cluster 2 identified as conservative, closely resemble the compositions reported by Newman in \cite{Newman2006}}
\end{center}
\end{table}

As is clear, increasing $d_{\text{thresh}}$ above $10$ makes little difference. Our results also agree with the clusters found by Newman in \cite{Newman2006}, where he finds one cluster which is $93\%$ liberal, and a second cluster which is $97\%$ conservative. Of course, to achieve this accuracy we have had to discard over $40\%$ of our data points. An interesting future line of research would be to improve the handling of low degree vertices by SCP.\\
 To detect clustering at a finer scale, we ran SCP to find a cluster of size $80$ containing a randomly selected liberal blog ({\tt smithersmpls.com}). The output of this is shown in figure \ref{fig:PolBlogsSmallScale}. As can be seen by the increased density of the top-left hand corner, our algorithm finds a subset of blogs containing {\tt smithersmpls.com} that are more densely connected to each other than they are to the rest of the data set. We remark that this experiment took approximately $0.06$ seconds, so with even modest computational resources one could investigate clustering at a large range of scales in such a data set.
\begin{figure}[H]
\begin{center}
\includegraphics[trim = 5mm 45mm 5mm 45mm, clip, width =0.75\textwidth]{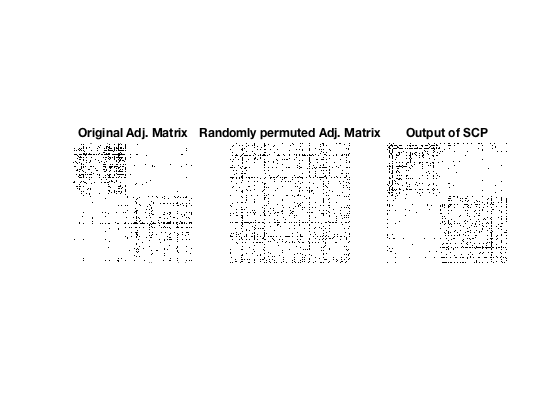}
\caption{Typical Output of SCP applied to the {\tt polblogs} data set with $d_{\text{thresh}} = 10$.}
\label{fig:PolBlogsLibCon}
\end{center}
\end{figure}

\begin{figure}[H]
\begin{center}
\includegraphics[trim = 0 0 0 15mm, width = 0.7\textwidth]{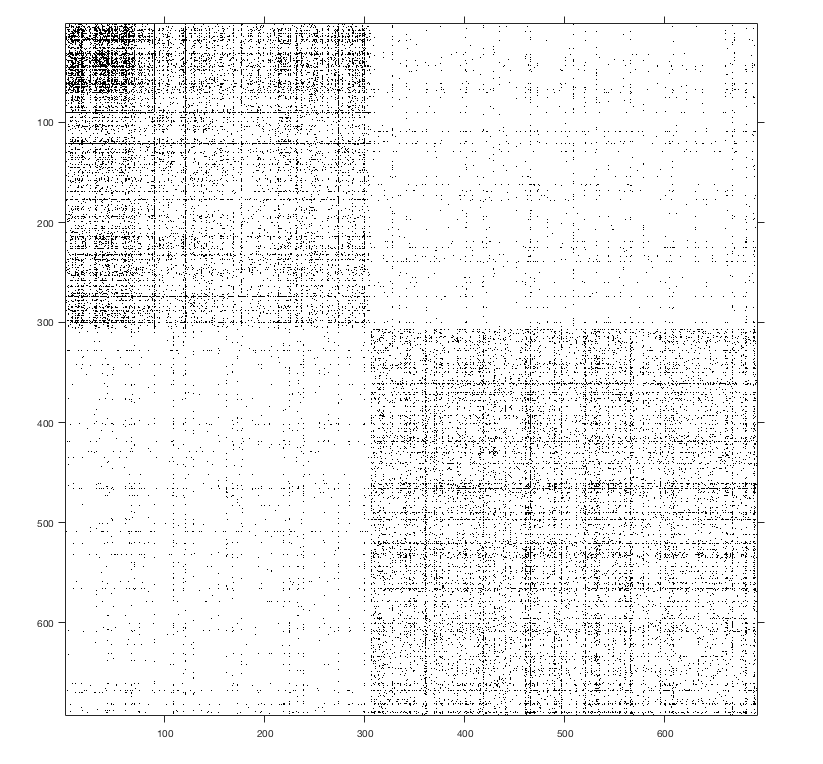}
\caption{Detecting fine scale structure in the {\tt polblogs} data set. The top left hand corner represents a sub-cluster of $80$ liberal blogs containing a given randomly selected liberal blog.}\label{fig:PolBlogsSmallScale}
\end{center}
\end{figure}

\subsubsection{Gene Expression Data Set}
\label{section:GeneExpression}
Our second data set consists of Gene Expression values collected via Microarray for a sample of \emph{Neurospora crassa} at $13$ different time points, originally studied in \cite{Dong2008}. The data consists of $2436$ time series, one for each Gene of interest, consisting of $13$ scaled expression readings. Each scaled expression reading is a floating point number between $4$ and $13$. We treat this data as a set of $2436$ data points $\bfX = \{\bfx_1,\ldots, \bfx_{2436}\}$ in $\mathbb{R}^{13}$. We then constructed an \emph{Affinity Matrix} $\bfA^1$ as suggested in \cite{Ng2001} defined as $\bfA^1_{ij} = \exp(-\|\bfx_{i} - \bfx_{j}\|^{2}_{2}/\sigma^2)$. Here $\sigma$ was chosen to be $\sqrt{10}$ although we note that experimenting with other values of $\sigma$ did not qualitatively change the results.\\

To promote sparsity, we added the additional step of constructing a $K$-nearest-neighbours adjacency matrix from $\bfA^1$ by retaining the (weighted) edge $\{i,j\}$ if and only if $\bfx_j$ is among the $K$ nearest neighbours of $\bfx_i$ or \emph{vice versa}. Call the resulting adjacency matrix $\bfA^{2,K}$. We remark that, much like the parameter $\sigma$ used in constructing $\bfA^1$, varying $K$ did not qualitatively affect the results obtained. Thus, we shall fix $K = 100$ and refer to $\bfA^{2,100}$ simply as $\bfA^{2}$. Figure \ref{fig:Gene_Expression_Before} shows the results of this preprocessing in reverse-grayscale (that is, larger values are darker).

\begin{figure}[H]
\begin{center}
\includegraphics[width=0.75\textwidth]{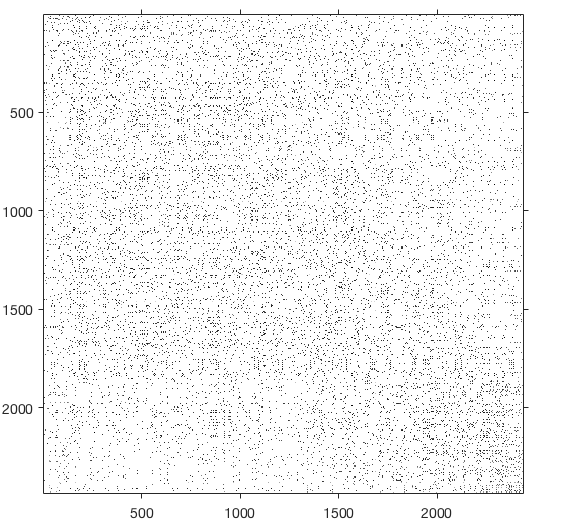}
\caption{The adjacency matrix $\bfA^2$, resulting from preprocessing the Gene Expression Data set originally studied in \cite{Dong2008}.}\label{fig:Gene_Expression_Before}
\end{center}
\end{figure}

We used ISCP to look for clusters on two different scales, informed by the underlying biology of the data set\footnote{and the authors gratefully acknowledge many informative discussions with Professor Jonathan Arnold and his lab}.  First, we looked for clusters of size $90 - 100$, as this is around the number of genes targeted by a single regulator in \emph{N. crassa}. The results of this experiment are shown in figure \ref{fig:Gene_Expression_Small_Cluster}.\\

In a well studied organism such as \emph{N. crassa}, the functions of many genes are known. In fact of the $2436$ genes in this data set, more than half of them ($1442$) have been assigned a MIPS label, which is an hierarchical system for labelling genes by function. Moreover, in \cite{Dong2008}, the genes in this data set are grouped into twelve categories according to their functions, for example `signaling' or `transcriptional control and regulation'. Biologically, an interesting question is which functional categories cluster together. Thus, for our second numerical experiment we used ISCP to detect clusters of size $450 - 550$ (results pictured in Figure \ref{fig:Gene_Expression_Large_Cluster}) and recorded the number of genes of each category \footnote{For a description of the categories the reader is referred to \cite{Dong2008}. The counts are displayed in Table \ref{table:Gene_Expression}. The categories are numbered by the order in which they appear in figure $15$ in \cite{Dong2008}, omitting the category `clock' as there are no MIPS codes associated with this category.} occurring in each cluster. We can assess the biological significance of this clustering by performing a chi-squared test. The null hypothesis is that there is no relation between the clusters and the functional categories, in which the expected number of genes in Category $j$ contained in Cluster $i$, in the notation of Table \ref{table:Gene_Expression}, is $E_{ij} = \mathbf{T}^{C}_{j} \times (\mathbf{T}^{R}_{i}/1442)$. For example, $E_{12} = 164 \times (321/1442) = 36.51$ Denoting the observed counts by $O_{ij}$ (i.e. $O_{12} = 32$, $O_{21} = 19$ and so on) we compute the chi-squared test statistic as $\sum_{i}\sum_{j} (E_{ij} - O_{ij})^2/E_{ij} = 73.87$. From a table of values for the $\chi^{2}$ distribution\footnote{Note that there are $(\# \text{rows} - 1)(\#\text{columns} - 1) = 40$ degrees of freedom here}, we get that assuming the null hypothesis there is a $0.1\%$ chance that this statistic is greater than $73.40$. Thus, we may safely reject the null hypothesis, and assume that the Clustering found by ISCP is related to the functions of the genes in the data set.

\begin{table}[h!]
\begin{center}
\scriptsize 
\begin{tabular}{|c|ccccccccccc|c|c|}
\hline 
\backslashbox{Cluster}{Category} & $\mathbf{1}$ & $\mathbf{2}$ & $\mathbf{3}$ & $\mathbf{4}$ & $\mathbf{5}$ & $\mathbf{6}$ & $\mathbf{7}$ & $\mathbf{8}$ & $\mathbf{9}$ & $\mathbf{10}$ & $\mathbf{11}$ & $\mathbf{T}^{R}$ & {\it Cluster size} \\
\hline \hline 
$\mathbf{1}$ & $13$ & $32$ & $15$ & $22$ & $42$ & $46$ & $\mathbf{36}$ & $\mathbf{19}$ & $8$ & $20$ & $68$ & $321$ & $527$ \\
\hline
$\mathbf{2}$ & $19$ & $36$ & $18$ & $28$ & $24$ & $37$ & $\mathbf{11}$ & $10$ & $12$ & $25$ & $70$ & $290$ & $513$ \\
\hline
$\mathbf{3}$ & $13$ & $\mathbf{23}$ & $21$ & $33$ & $34$ & $\mathbf{57}$ & $\mathbf{10}$ & $13$ & $11$ & $32$ & $71$ & $318$ & $452$\\
\hline  
$\mathbf{4}$ & $15$ & $32$ & $14$ & $17$ & $22$ & $27$ & $17$ & $\mathbf{2}$ & $5$ &$22$ & $47$ & $220$ & $554$ \\
\hline
$\mathbf{5}$ & $20$ & $41$ & $19$ & $31$ & $35$ & $\mathbf{26}$ & $18$ & $15$ & $9$ & $\mathbf{13}$ & $66$ & $293$ & $390$ \\
\hline\hline
$\mathbf{T}^{C}$ &  $80$ & $164$ & $87$ & $131$ & $157$ & $193$ & $92$ & $59$ & $45$ & $112$& $322$ & $1442$ & $2436$ \\
\hline
\end{tabular}
\caption{The number of genes in each category present in each cluster. The row totals represent the number of genes with a MIPS classifier present in each cluster (the total number of genes with a MIPS classifier is $1442$). Anomalously high or low counts are shown in bold.}
\end{center}
\label{table:Gene_Expression}
\end{table}

\begin{figure}[H]
\begin{center}
\includegraphics[width=0.65\textwidth]{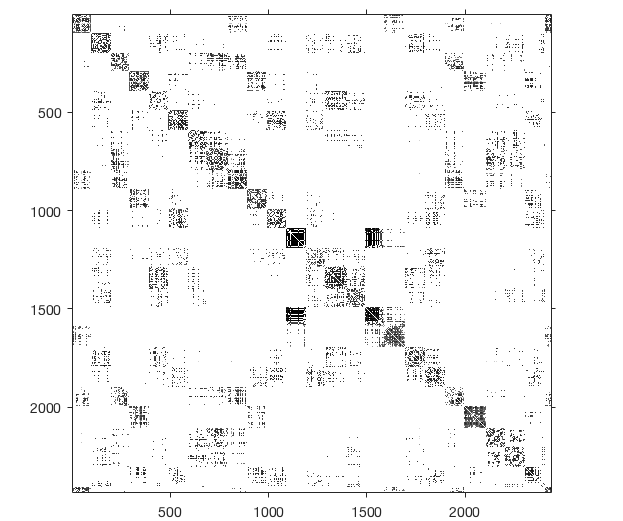}
\caption{Output of ISCP applied to $\bfA^2$, set to find clusters of size $90-100$.}\label{fig:Gene_Expression_Small_Cluster}
\end{center}
\end{figure}

\begin{figure}[H]
\begin{center}
\includegraphics[width=0.65\textwidth]{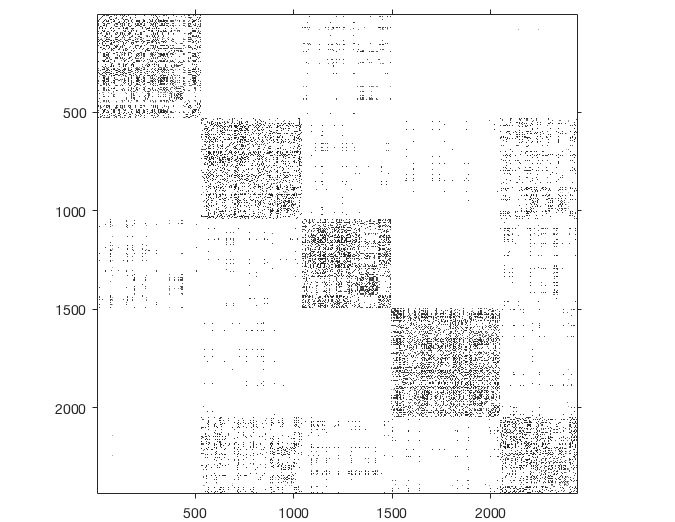}
\caption{Output of ICSP applied to $\bfA^2$, set to find clusters of size $450-550$.}\label{fig:Gene_Expression_Large_Cluster}
\end{center}
\end{figure}

\end{document}